\documentclass[a4paper,11pt]{article}  % add "twoside" for page waving"

\def\refcolour{green!60!black}
\def\citecolour{cyan}
\usepackage{algorithm}
\usepackage[noend]{algpseudocode}
\algnewcommand{\algorithmicassumption}{\textbf{Requirement:}}
\algnewcommand{\Assume}{\item[\algorithmicassumption]}
\algrenewcomment[1]{\hfill \(\triangleright\) \emph{\small #1}} % to italicize
\algnewcommand{\InlineIf}[2]{% single line if-then
  \algorithmicif\ #1\ \algorithmicthen\ #2}
\algnewcommand{\InlineElse}[1]{% single line else
  \algorithmicelse\ #1}
\algnewcommand{\InlineIfElse}[3]{% single line if-then-else
  \algorithmicif\ #1\ \algorithmicthen\ #2\ \algorithmicelse\ #3}
\algnewcommand{\InlineFor}[2]{\algorithmicfor\ #1\ \algorithmicdo\ #2}
\algnewcommand{\CommentLine}[1]{\(\triangleright\) \emph{\small #1}}

\algnewcommand{\algorithmicand}{\textbf{and}}
\algnewcommand{\algorithmicor}{\textbf{or}}
\algnewcommand{\FOR}{\algorithmicfor}
\algnewcommand{\OR}{\algorithmicor}
\algnewcommand{\AND}{\algorithmicand}
\algnewcommand{\IF}{\algorithmicif}
\algnewcommand{\THEN}{\algorithmicthen}
\algnewcommand{\ELSE}{\algorithmicelse}
\algnewcommand{\True}{\textsc{True}}
\makeatletter
\newcommand{\StateX}[1]{%
  \setlength\@tempdima{\algorithmicindent}%
  \Statex\hskip\dimexpr#1\@tempdima\relax}
\makeatother
\makeatletter
\newcommand{\algoCaptionLabel}[3]{
     \caption[\textproc{#1}]{\textproc{#1}\ifthenelse{\equal{#2}{}}{}{$(#2)$} }%
     \NR@gettitle{\textproc{#1}}%
      \label{algo:#3}
     }%
\makeatother

\usepackage[utf8]{inputenc}
\usepackage[T1]{fontenc}
\usepackage{pgfplots}
\pgfplotsset{compat=newest}
\usepgfplotslibrary{fillbetween}
\usetikzlibrary{patterns}
\usetikzlibrary{arrows.meta}    % required for custom arrowheads in tikz
\usepackage{amsmath}
\usepackage{amssymb}
\usepackage{mathtools}	    % for \coloneqq and \prescript
\usepackage[hypertexnames=false]{hyperref}
\usepackage{a4wide}
\usepackage{array}
\usepackage{titling}
\usepackage{setspace}
\usepackage{geometry}       % [left=2cm,right=2cm,top=2cm,bottom=3cm]
\hypersetup{				% for colours of clickable stuff
    colorlinks,				% show colours
    citecolor=\citecolour,	% colour for \cite{} (NOT brackets, number only)
    filecolor=red,			% colour for internal file reference
    linkcolor=\refcolour,	% colour for \ref{} and \nameref{}
    urlcolor=black			% colour for \url{}
}
\usepackage{comment}
\usepackage{amsthm}
\usepackage[capitalize,nameinlink]{cleveref}	% for smart clickable stuff
\usepackage{bm}			% for bold math symbols
\usepackage{mathrsfs}
\usepackage[shortlabels]{enumitem}
\usepackage[scr=boondoxo]{mathalpha}
\usepackage{anyfontsize}
\usepackage{float}
\usepackage{graphicx}
\usepackage[indent]{parskip}
\usepackage{multirow}
\usepackage{authblk}

\usepackage{etoolbox}

\makeatletter

\patchcmd\deferred@thm@head
  {\addvspace{-\parskip}}
  {}
  {}{\typeout{\string\deferred@thm@head patch failed!}}

\makeatletter 

\newtheorem{theorem}{Theorem}[subsection]

\renewcommand*{\thetheorem}{%
  \ifnum\value{subsection}<1 % if `subsection` counter is less than 1
    \thesection% show \thesection, e.g., 1
  \else% otherwise
    \thesubsection% show \thesubsection, e.g., 1.1
  \fi
  .\arabic{theorem}%
}

\newtheorem{lemma}[theorem]{Lemma}
\AddToHook{env/lemma/begin}{\crefalias{theorem}{lemma}}
\Crefname{lemma}{Lemma}{Lemmas}
\newtheorem{corollary}[theorem]{Corollary}
\AddToHook{env/corollary/begin}{\crefalias{theorem}{corollary}}
\Crefname{corollary}{Corollary}{Corollaries}

\AddToHook{env/definition/begin}{\crefalias{theorem}{definition}}
\Crefname{definition}{Definition}{Definitions}
\newtheorem{proposition}[theorem]{Proposition}
\AddToHook{env/proposition/begin}{\crefalias{theorem}{proposition}}
\Crefname{proposition}{Proposition}{Propositions}
\AddToHook{env/example/begin}{\crefalias{theorem}{example}}
\Crefname{example}{Example}{Examples}
\newtheorem{remark}[theorem]{Remark}
\AddToHook{env/remark/begin}{\crefalias{theorem}{remark}}
\Crefname{remark}{Remark}{Remarks}

\newtheorem{innercustomgeneric}{\customgenericname}
\providecommand{\customgenericname}{}
\newcommand{\newcustomtheorem}[2]{%
  \newenvironment{#1}[1]
  {%
   \ifdefined\crefalias\crefalias{innercustomgeneric}{#2}\fi
   \renewcommand\customgenericname{#2}%
   \renewcommand\theinnercustomgeneric{##1}%
   \innercustomgeneric
  }
  {\endinnercustomgeneric}%
  \ifdefined\crefname\crefname{#2}{#2}{#2s}\fi
}

\newcustomtheorem{customthm}{Theorem}

\allowdisplaybreaks

\crefname{equation}{}{} % Makes \Cref{equationline} write "(#)" over "Eq.(#)"

\makeatletter
\renewcommand\@makefnmark{
  \hbox{\@textsuperscript{\normalfont\color{blue}\@thefnmark}}}
\makeatother

\DeclarePairedDelimiter{\norm}{\lVert}{\rVert}
\DeclarePairedDelimiter{\abs}{\lvert}{\rvert}
\DeclareMathOperator{\htt}{ht}
\DeclareMathOperator{\Jac}{Jac}
\DeclareMathOperator{\GL}{GL}

\DeclareMathOperator{\crit}{crit}

\DeclareMathOperator{\grad}{grad}

\DeclareMathOperator{\rank}{rank}

\DeclareMathOperator{\I}{I}
\DeclareMathOperator{\T}{T}
\DeclareMathOperator{\V}{V}

\DeclareMathOperator{\der}{d}

\newcommand{\bma}{\bm{a}}
\newcommand{\bmb}{\bm{b}}

\newcommand{\bmd}{\bm{d}}
\newcommand{\bme}{\bm{e}}
\newcommand{\bmn}{\bm{n}}
\newcommand{\bmp}{\bm{p}}
\newcommand{\bmq}{\bm{q}}
\newcommand{\bmr}{\bm{r}}
\newcommand{\bms}{\bm{s}}
\newcommand{\bmu}{\bm{u}}
\newcommand{\bmv}{\bm{v}}
\newcommand{\bmx}{\bm{x}}
\newcommand{\bmX}{\bm{X}}
\newcommand{\bmy}{\bm{y}}
\newcommand{\bmz}{\bm{z}}
\newcommand{\bmzero}{\bm{0}}
\newcommand{\balpha}{{\bm{\alpha}}}
\newcommand{\betta}{{\bm{\beta}}}
\newcommand{\bmeta}{{\bm{\eta}}}
\newcommand{\bxi}{{\bm{\xi}}}
\newcommand{\bell}{{\bm{\ell}}}

\newcommand{\field}{\mathbb{K}}
\newcommand{\Nintegers}{\mathbb{N}}
\newcommand{\Zintegers}{\mathbb{Z}}
\newcommand{\ratfield}{\mathbb{Q}}
\newcommand{\realfield}{\mathbb{R}}
\newcommand{\compfield}{\mathbb{C}}

\newcommand{\assS}{\textbf{(S)}}
\newcommand{\OGassumS}{\hypertarget{S}{\textcolor{\refcolour}{\assS}}}
\newcommand{\assumS}{\hyperlink{S}{\assS}}

\newcommand{\assB}{\textbf{(B)}}
\newcommand{\OGassumB}{\hypertarget{B}{\textcolor{\refcolour}{\assB}}}
\newcommand{\assumB}{\hyperlink{B}{\assB}}

\newcommand{\assHa}{\textbf{(A1)}}
\newcommand{\OGassumHa}{\hypertarget{H1}{\textcolor{\refcolour}{\assHa}}}
\newcommand{\assumHa}{\hyperlink{H1}{\assHa}}
\newcommand{\assHb}{\textbf{(A2)}}
\newcommand{\OGassumHb}{\hypertarget{H2}{\textcolor{\refcolour}{\assHb}}}
\newcommand{\assumHb}{\hyperlink{H2}{\assHb}}

\newcommand{\ttRatParam}{\texttt{ZeroDimParam}}

\newcommand{\ttApproximationList}{\texttt{ParamApproximation}}
\newcommand{\ttParamMaxValue}{\texttt{UpperBoundValue}}
\newcommand{\ttParamMinValue}{\texttt{LowerBoundValue}}

\newcommand{\ttSLPPolar}{\texttt{Polar}}

\newcommand{\ttExtensionRealRoots}{\texttt{ExtensionIsolation}}

\newcommand{\procCrit}{\textproc{CriticalPoints}}
\newcommand{\procApprox}{\textproc{Approximation}}
\newcommand{\procIsol}{\textproc{Isolation}}

\newcommand{\GamSLP}{\Gamma}
\newcommand{\LamSLP}{\Lambda}

\newcommand{\scrQ}{\mathscr{Q}}
\newcommand{\scrM}{\mathscr{M}}
\newcommand{\scrm}{\mathscr{m}}
\newcommand{\scrc}{\mathscr{c}}
\newcommand{\scrd}{\bm{\mathscr{d}}}
\newcommand{\scrA}{\mathscr{A}}
\newcommand{\scrI}{\mathscr{I}}
\newcommand{\scrC}{\mathscr{C}}
\newcommand{\scrP}{\mathscr{P}}
\newcommand{\scrB}{\mathscr{B}}
\newcommand{\scrH}{\mathscr{H}}
\newcommand{\scrO}{\mathscr{O}}

\newcommand{\scrN}{\mathscr{N}}

\newcommand{\setD}{X}
\newcommand{\setE}{Y}

\newcommand{\epsproba}{\epsilon}
\newcommand{\changeofvarmatrix}{\mathbf{A}}

\newcommand{\genericcoordsigma}{\bm{\sigma}}

\newcommand{\critpoints}{\text{crit}}
\newcommand{\probaP}{\mathbb{P}}
\newcommand{\approxheight}{\gamma}

\newcommand{\totaldegree}{\bm{D}}
\newcommand{\multitotaldegree}{\bm{\Delta}}
\newcommand{\multitotalheight}{\bm{H}}

\newcommand{\calB}{\mathcal{B}}

\newcommand{\calP}{\mathcal{P}}

\newcommand{\calV}{\mathcal{V}}
\newcommand{\calZ}{\mathcal{Z}}
\newcommand{\frakA}{\mathfrak{A}}

\newcommand{\frakS}{\mathfrak{S}}

\newcommand{\varT}{t}
\newcommand{\varU}{u}

\newcommand{\upMat}{C}
\newcommand{\lowMat}{B}
\newcommand{\lowEnt}{b}

\makeatletter
\newcommand{\vast}{\bBigg@{3}}
\newcommand{\Vast}{\bBigg@{4}}
\makeatother
\floatstyle{ruled}
\newfloat{subroutine}{htbp}{lok}
\floatname{subroutine}{Subroutine}

\crefname{subroutine}{subroutine}{subroutines}
\Crefname{subroutine}{Subroutine}{Subroutines}

\makeatletter
\newcommand{\subrCaptionLabel}[3]{
    \caption[\textproc{#1}]{\textproc{#1}\ifthenelse{\equal{#2}{}}{}{$(#2)$} }%
    \NR@gettitle{\textproc{#1}}%
      \label{subr:#3}
    }%
\makeatother

\makeatletter
\AddToHook{env/algorithmic/begin}{\def\@currentcounter{ALG@line}}
\makeatother

\crefalias{ALG@line}{line}

\title{Computing points in connected components defined by a real inequation: 
algorithms, complexity and implementations, Part I}
\author[1]{J{\'e}r{\'e}my Berthomieu}
\author[1]{Edern Gillot}
\author[1]{Mohab {Safey El Din}}
\affil[1]{Sorbonne Université, CNRS, LIP6, Paris, France}
\date{\today}
\begin{document}
%\pagenumbering{gobble}

\maketitle

\begin{abstract}
    We consider the problem of computing sample points in each connected component
of a semi-algebraic set defined by the non-vanishing or the positivity of an
$n$-variate polynomial of degree $d$, with rational coefficients of bit size
bounded by $\tau$. 

Such a problem is a basic routine in effective real algebraic geometry, used
in higher-level algorithms for solving polynomial systems over the reals and
finds many applications in sciences. 

We design a probabilistic algorithm for solving this problem, which is based on
reductions to different routines for solving zero-dimensional polynomial
systems. It assumes that the input polynomial satisfies sufficiently generic
properties (namely, smoothness of its defining hypersurface). This is done
through the computations of critical points of well-chosen maps to capture the
connected components of the semi-algebraic set under study.

We derive a bit complexity estimate for the cost of this algorithm, which is, in
terms of the B\'ezout bound $d(d-1)^{n-1}$, essentially cubic for obtaining
parametrisations of the sought-for real points. Moreover, we also consider the
case of obtaining rational approximations of those points, which are precise
enough to lie in the same connected components as their exact counterparts,
which yields a cost that is essentially quartic in the B\'ezout bound. In these
complexity estimates, we take into account the degree structure of the input
polynomial and its partial derivatives, allowing for a more refined bit
complexity when the partial derivative of the input polynomial have degree lower
than expected. We also analyse the probability of success of those algorithms.

We report on practical experiments, benchmarking with random dense
input polynomials as well as polynomials coming from applications, which were
out of reach of the state-of-the-art implementations, and hence illustrate the
practical efficiency of these new algorithms.
\end{abstract}

\section{Introduction}\label{sec:intro}

\subsection{Background}

% Understanding the set of solutions, and particularly real solutions, to
% polynomial systems with constraints arise in a wide range of applications, as
% such systems encode non-linear algebraic relations, and are hence used to model
% many physical and biological properties. However, answering even basic queries
% about those sets is difficult; for instance, deciding if such a set is empty or
% not is shown to be an NP-hard problem \cite[Appendix A.7]{GS79}. Nonetheless,
% these sets enjoy remarkable properties, such as the finiteness of their
% connected components \cite{whitney,loja}.

Computing sample points of connected components of real solutions to real
polynomial systems with constraints is a basic subroutine of real algebraic
geometry. It allows in particular to decide the emptiness of the solution set of
such a system, a problem that is known to be NP-hard \cite[Appendix A.7]{GS79}.
The computation of these sample points also finds applications in fields such as
robotics \cite{CPSSW22}, biology \cite{FS22}, optimisation
\cite{DBLP:phd/hal/Ferguson22}, program verification \cite{GHMM23}, or
combinatorics \cite{IS24}. There is therefore a need for efficient algorithms,
both in theory and in practice, for computing these sample points. Moreover, as
examples arising from applications are typically structured, having for instance
symmetries, invariance, multi- or weighted-homogeneity, these algorithms should
make use of these structures to improve computations. In this paper, we focus on
the particular case of semi-algebraic sets defined by a single polynomial
inequation: sets of the form $S \coloneqq \left\{\bmx \in \mathbb{R}^n \ | \
f(\bmx) \neq 0\right\}$, where $f$ is a polynomial with rational coefficients in
$n$ variables.

Given such a polynomial $f$ of degree $d$, the number of connected components of
$S$ is known to be at most $O(d)^n$, by the
\emph{Oleinik--Petrovski--Thom--Milnor bound} (original works
in \cite{OP49,thom,milnor}, see \cite{BPR} for a modern statement). In
particular, any algorithm computing at least one point per connected component
of $S$ will have complexity at least exponential in $n$.
% Actual bound is \binom{s}{dim} O(d)^n, where s = number of ineq. and dim =
% dimension of solution space of the equalities (see ON THE BETTI NUMBERS OF
% SIGN CONDITIONS) by BPR. More usable bound (slightly worse) is s^k O(d)^k
% where s is the number of inequalities

This problem can be seen as an instance of \emph{real quantifier elimination},
and as such is known to be solvable algorithmically since Tarski \cite{tarski}.
The first algorithm whose complexity is elementary recursive to solve this
problem is the \emph{cylindrical algebraic decomposition} (CAD) given by Collins
\cite{collins}, and performs $d^{2^{O(n)}}$ arithmetic operations. Moreover, if
the bitsize of the coefficients of $f$ is at most $\tau$, the CAD performs $\tau
d^{2^{O(n)}}$ bit operations \cite{BPR,BHPSS21}. The CAD is a classical
algorithm for solving quantifier elimination problems, and is available in
various computer algebra systems, such as \texttt{Maple} \cite{maple} and
\texttt{Mathematica} \cite{mathematica}. However, its doubly exponential
complexity in $n$ makes it impractical to use when $n$ is large; for instance, 
on examples of fixed degree $4$ and small bitsize, its runtime becomes
prohibitive whenever $n$ exceeds $4$.

This doubly exponential complexity in the number of variables contrasts with the
Oleinik--Petrovski--Thom--Milnor bound, which is singly exponential in $n$. To
bridge this gap, the \emph{critical point} method was developed
\cite{GV88,GV92,HRS94}, and eventually resulted in Basu, Pollack and Roy's
algorithm \cite[Theorem 13.22]{BPR}, which performs $d^{O(n)}$ arithmetic
operations, and $\tau d^{O(n)}$ bit operations, and is thus polynomial in the
Oleinik--Petrovski--Thom--Milnor bound. However, the constant hidden by the
big-Oh notation is actually too large for practical implementation \cite{ARS02},
due to the manipulation of infinitesimals, and it does not make use of any
structure inherent to the input system.

To resolve this, new, \emph{probabilistic}, algorithms have been developed
\cite{BGHM97,SS03,din07,BGHLMS15,LeS21,LeS22}, where a random change of
variables is performed to ensure certain properties with high probability.
On one hand, one could reduce the problem of computing sample points in a
semi-algebraic set to one in an \textit{algebraic} set, by introducing a new
variable $x_{n+1}$ and taking $x_{n+1} f - 1$ as input. In this case, using
$\widetilde{O}$ to denote the omission of logarithmic factors, the
state-of-the-art complexity result is in
$\widetilde{O}\big(\log(1/\epsproba)(\tau + \log(1/\epsproba))d^{3n+6}\big)$ bit
operations \cite{SS03,EGS20,EGS23}, where $0 < \epsproba < 1$ is an input
parameter controlling the probability of success. Note that this result is only
applicable on inputs satisfying some regularity assumption.
On the other hand, algorithms specific to the semi-algebraic case have also been
developed, yielding a state-of-the-art complexity result of
$\widetilde{O}\big(\binom{n+d}{d} 8^n d^{2n+1}\big)$ \textit{arithmetic}
operations \cite{LeS22}, without analysis of the bit complexity.
Both these approaches enjoy better practical results than the CAD.

In the following paragraphs, we recall the notion of height (a measure for the
bit size of polynomials), as well as certain data structures used to represent
polynomials with rational coefficients and finite sets of algebraic points,
which are needed to express our main result.

\subsection{Bit size and data structures}

Throughout this paper, we use $\log$ to denote logarithms in base $2$.

\paragraph*{Height.} We define the (logarithmic) \emph{height} of a non-zero
rational number $q = u/v$ as $\htt(q) \coloneqq \max(\log \abs{u}, \log(v))$,
where $u \in \Zintegers$, $v \in \Nintegers \setminus \{0\}$, and $\gcd(u,v)=1$,
and formally define $\htt(0) \coloneqq 0$. For a non-zero polynomial $g \in
\ratfield[x_1,\dots,x_n]$, if $v \in \Nintegers \setminus \{0\}$ is the minimal
common denominator of all non-zero coefficients of $g$, then the \emph{height}
of the polynomial $g$ is defined as $\htt(g) \coloneqq \max(\log(v),
\log\norm{vg}_\infty)$, where $\norm{\cdot}_\infty$ denotes the maximum of the
absolute values of the coefficients of the polynomial. Knowing the height of a
polynomial with rational coefficients, along with its degree, allows us to
obtain an upper bound on the size of its bit representation.

\paragraph*{Straight-line programs.} Given polynomials $g_1, \dots, g_r \in
\ratfield[x_1, \dots, x_n]$, we say that a sequence $\GamSLP$ of elementary
operations $+$,$-$,$\times$ is a \emph{straight-line program evaluating the
polynomials} $g_1, \dots, g_r$ if $\GamSLP$ evaluates $g_1, \dots, g_r$ from the
input variables $x_1, \dots, x_n$. The \emph{length}, or \emph{size}, of a
straight-line program is the number of elementary operations it performs. This
representation of polynomials is commonly used in polynomial system solving
\cite{GHMP95,GHHMMP97,BGHM97,GHMMP98,GLS01,SS03,JS07,SS18}. 

\paragraph*{Zero-dimensional rational parametrisations.} Suppose that we are
given a zero-dimensional algebraic set $Q \subset \overline{\ratfield}^n$
defined over $\ratfield$ (that is, a finite set of algebraic numbers). A
\emph{zero-dimensional rational parametrisation} $\scrQ = (w, v_1, \dots, v_n,
\nu)$ of $Q$ consists of univariate polynomials $w,v_1,\dots,v_n \in
\ratfield[\varT]$ for a new variable $\varT$, and a $\ratfield$-linear form
$\nu$ in $n$ variables, such that:
\begin{itemize}
    \item $w$ is monic, squarefree, and $\deg(w) > \deg(v_i)$ for all $i$,
    \item $\displaystyle \nu(v_1, \dots, v_n) \equiv \varT w' \mod w$,
    \item and \(\displaystyle Q = \left\{\left(\frac{v_1(\varT)}{w'(\varT)}, 
    \dots, \frac{v_n(\varT)}{w'(\varT)}\right) \in \compfield^n \ \middle| 
    \ w(\varT) = 0\right\}\).
\end{itemize}
The constraint on $\nu$ says that the roots of $w$ are precisely the values
taken by $\nu$ on the points of $Q$, and implies that $\nu$ takes distinct
values at distinct points of $Q$; it is therefore referred to as a
\emph{separating} linear form \emph{associated} to $\scrQ$. This data structure
is a staple of effective algebra, being first introduced by Kronecker and
Macaulay \cite{kronecker,macaulay} and having been used in many algorithms since
\cite{GM89,GHMP95,ABRW96,GHMMP98,rouillier99,GLS01,SS03,PS13,SS18}.

Note that, although the denominators $w'$ are not strictly necessary to
represent the sets in this way, they allow for precise control of the heights of
the polynomials of the parametrisation \cite{ABRW96,GLS01,rouillier99}, which we
exploit in this work. By convention, we denote by $\scrQ = (w \equiv 1)$ the
parametrisation that represents the empty set. Moreover, we define the
\textit{degree} of a parametrisation as the degree of its defining polynomial
$w$.

\subsection{Main result}

This work is split into two parts. Our first contribution to the problem,
covered in this paper, is a new Monte-Carlo algorithm for computing points per
connected components of semi-algebraic sets defined by a single polynomial
inequation, in the case where the hypersurface defined by the vanishing of this
polynomial is smooth. Our second contribution will be a similar algorithm, which
does not require any smoothness assumption.

The smooth-case algorithm relies on the use of efficient multi-homogeneous
polynomial system solving techniques \cite{SS18} for solving zero-dimensional
polynomial systems, and thus inherits their probabilistic aspects. As such, our
algorithm implicitly relies on an \emph{oracle} $\scrO$, which takes as input a
positive integer $b$, and returns a prime number in the set $\{b+1, \dots, 2b\}$
uniformly at random among the primes in this interval. Such an oracle exists,
see for example \cite[Section 18.4]{MCA}.

We adapt the algorithm from \cite{SS03} to compute points on the boundaries of
each connected component of the semi-algebraic set $S$ under study, and obtain
from them parametrisations of finite sets of points meeting all connected
components of $S$, by computing well-chosen translations of these critical
points. Once the parametrisations of the sample points are obtained, it
suffices to compute rational approximations of those points of low bit length,
but sufficiently precise to lie in the same connected component as their exact
counterpart. This idea is not new (see e.g.\ \cite[Theorem 13.16]{BPR}), but we
apply recent results concerning rational parametrisations (chiefly due to
\cite{MS21}) in order to make explicit the exponent of $d$.

Concerning the computation of these rational approximations, our first result
concerns the case where the input polynomial $f$ has a general multi-homogeneous
structure:

\begin{customthm}{1}\label{thm:mainmulti}
    Suppose that $f \in \ratfield[\bmx_1, \dots, \bmx_m]$ is squarefree, where
    $\bmx_i = \Big(x_{1 +\sum_{j=1}^{i-1}n_j}, \dots$, $x_{n_i +
    \sum_{j=1}^{i-1}n_j}\Big)$, $n = n_1 + \dots + n_m$, and $\deg_{\bmx_i}(f)
    \leq d_i$. Suppose that $d_1 \leq \cdots \leq d_m$, and that $f$ has height
    $\tau$. Assume that the hypersurface $\V(f) \subset \compfield^n$ is smooth,
    and that $\GamSLP$ is a straight-line program of length $L$ evaluating $f$.
    Let $S \coloneqq \left\{\bmx \in \mathbb{R}^n \ | \ f(\bmx) \neq 0\right\}$
    and let $0 < \epsproba < 1$. 
    Then there exists an algorithm which takes $\GamSLP$ and $\epsproba$ as
    inputs, and returns one of the following outputs:
    \begin{itemize}
        \item either a set of at most $2n\multitotaldegree + 1$ points of
        $\ratfield^n$,
        \item or \emph{\texttt{fail}}.
    \end{itemize}
    With probability at least $1-\epsproba$, the first outcome occurs, and the
    computed set of points intersects every connected component of $S$. In any
    case, the algorithm performs at most
    \[\widetilde{O}\big(\log(1/\epsproba)n^3\multitotaldegree
    \multitotalheight(L+nd+n^3) + nd^5\multitotaldegree^3(\multitotalheight 
    + n\multitotaldegree)\big)\]
    bit operations, where:
    \begin{itemize}
        \item $d \coloneqq \deg(f)$,
        \item for $1 \leq c \leq m$, $1 \leq k \leq n$, $\ell > k$, and
        constants $a_{i,j}$ and $b_{i,j,k}$, 
        \[\delta_{c,\ell,k} \coloneqq \deg_{{\bmx_c}}\bigg(\frac{\partial
        f}{\partial x_{\ell}} + \sum_{i=1}^k
        \bigg(\sum_{j=1}^{n-k}a_{i,(j+k)}\lowEnt_{j,\ell-k,k}\bigg)
        \frac{\partial f}{\partial x_i}\bigg),\]
        \item $\approxheight \in \widetilde{O}\left(n\log(1/\epsproba) + n^2 + 
        \tau + d\right)$,
        \item $\multitotaldegree$ is the sum of the coefficients of the
        polynomial
        \[(d_1\theta_1 + \dots + d_m\theta_m)\prod_{i=2}^{n}
        \left(\delta_{1,i,1}\theta_1 + \dots + \delta_{m,i,1}\theta_m\right)
        \mod \left<\theta_1^{n_1+1}, \dots, \theta_m^{n_m+1}\right>,\]
        which is bounded from above by $d_1^{n_1}\cdots d_m^{n_m}
        \frac{n!}{n_1!\cdots n_m!}$,
        \item $\multitotalheight$ is the sum of the coefficients of the
        polynomial
        \begin{align*}
            &(\approxheight\zeta + d_1\theta_1 + \dots + d_m\theta_m)
            \prod_{i=2}^n \left(\approxheight\zeta + \delta_{1,i,1}\theta_1 + 
            \dots + \delta_{m,i,1}\theta_m\right) \\
            &\mod \left<\zeta^2, \theta_1^{n_1+1}, \dots, \theta_m^{n_m+1}
            \right>.
        \end{align*}
    \end{itemize}
    The oracle $\scrO$ is called $n$ times. Upon success, every
    rational coordinate in the output has height in
    \[\widetilde{O}\big(d^5\multitotaldegree^2(\multitotalheight +
    n\multitotaldegree)\big).\] 
\end{customthm}

Note that the squarefreeness of $f$ is a mild assumption, as the non-vanishing
of a polynomial defines the same semi-algebraic set as the non-vanishing of its
squarefree part. As $\multitotaldegree \leq d_1^{n_1}\cdots d_m^{n_m}
\frac{n!}{n_1!\cdots n_m!}$, the complexity result stated above can be
subexponential in $n$ in some cases, for instance when some of the $d_i$'s are
equal to $1$.

Our second result is a consequence of the first one, when the input polynomial
$f$ does not have any particular associated multi-degree structure.

\begin{customthm}{2}\label{thm:mainresult}
    Suppose that $f \in \ratfield[x_1, \dots, x_n]$ is squarefree, of degree $d$
    and height $\tau$, and that \(\deg\left(\frac{\partial f}{\partial
    x_1}\right) \leq \dots \leq \deg\left(\frac{\partial f}{\partial
    x_n}\right)\). Suppose that the hypersurface $\V(f) \subset \compfield^n$ is
    smooth. Let $S \coloneqq \left\{\bmx \in \mathbb{R}^n \ | \
    f(\bmx) \neq 0\right\}$. Suppose that $\GamSLP$ is a straight-line
    program of length $L$ evaluating $f$. Let $0 < \epsproba < 1$, and define
    \[\totaldegree \coloneqq d \times \deg\left(\frac{\partial f}{\partial
    x_2}\right) \times \dots \times \deg\left(\frac{\partial f}{\partial
    x_n}\right).\]
    Then there exists an algorithm which takes $\GamSLP$ and $\epsproba$ as
    inputs, and returns one of the following outputs:
    \begin{itemize}
        \item either a set of at most $2n\totaldegree + 1$ points of
        $\ratfield^n$,
        \item or \emph{\texttt{fail}}.
    \end{itemize}
    With probability at least $1-\epsproba$, the first outcome occurs, and the
    computed set of points intersects every connected component of $S$. In any
    case, the algorithm performs at most
    \[\widetilde{O}\big(n^2\totaldegree^2(n\log(1/\epsproba) + n^2 + 
    \tau + d)(\log(1/\epsproba)n^2(L+nd+n^3) + 
    d^5\totaldegree^2)\big)\]
    bit operations. The oracle $\scrO$ is called $n$ times. Upon success, every
    rational coordinate in the output has height in
    \[\widetilde{O}\big(nd^5\totaldegree^3(n\log(1/\epsproba) + n^2 + 
    \tau + d)\big).\]
\end{customthm}

The assumption on the degrees of the partial derivatives of $f$ is also mild,
since it can always be satisfied by re-labelling the variables.

Our third main result concerns the complexity of computing only the rational
parame\-trisations of the points per connected components of $S$, without
computing the rational approximations, which yields the following result:

\begin{customthm}{3}\label{cor:mainresult}
    Under the notations of \emph{\Cref{thm:mainmulti}} and
    \emph{\Cref{thm:mainresult}}, there exists an algorithm which takes
    $\GamSLP$ and $\epsproba$ as inputs, and returns one of the following
    outputs:
    \begin{itemize}
        \item either a point of $\ratfield^n$ and $2n$ zero-dimensional rational
        parametrisations,
        \item or \emph{\texttt{fail}}.
    \end{itemize}
    With probability at least $1-\epsproba$, the first outcome occurs, and the
    union of the point and the zeroes of the parametrisations contains at least
    one point per connected component of $S$. In any case, the algorithm
    performs at most
    \[\widetilde{O}\big(\log(1/\epsproba)n^3\multitotaldegree
    \multitotalheight(L+nd+n^3) + nd^3\multitotaldegree^2(\multitotalheight 
    + n\multitotaldegree)\big)\]
    bit operations in the multi-homogeneous setting, or 
    \[\widetilde{O}\big(n^2\totaldegree^2(n\log(1/\epsproba) + n^2 + 
    \tau + d)(\log(1/\epsproba)n^2(L+nd+n^3) + 
    d^3\totaldegree)\big)\]
    bit operations in the setting of \emph{\Cref{thm:mainresult}}. 
    The oracle $\scrO$ is called $n$ times. Upon success, every
    polynomial in each parametrisation has degree at most $\multitotaldegree$,
    respectively $\totaldegree$, and height in 
    \(\widetilde{O}\big(d^3\multitotaldegree (\multitotalheight +
    n\multitotaldegree)\big)\), respectively \(\widetilde{O}\big(nd^3
    \totaldegree^2 (n\log(1/\epsproba) + n^2 + \tau + d)\big)\).
\end{customthm}

Moreover, the quantities $\multitotaldegree$ and $\totaldegree$ are
(multi-homogeneous) B\'ezout bounds on the number of complex solutions to a
polynomial system involving $f$ and its partial derivatives. As such, these
quantities are bounded above by $d(d-1)^{n-1}$, and
\Cref{thm:mainmulti,thm:mainresult,cor:mainresult} imply that the bit complexity
of computing rational parametrisations of points per connected components of $S$
is essentially cubic in $d(d-1)^{n-1}$, whereas the bit complexity of computing
the actual rational sample points is essentially quartic in $d(d-1)^{n-1}$.

% Should one happen to know additional information about this system
% that allows to determine a better bound, such as a multi-homogeneous
% structure for a multi-homogeneous bound \cite{Waerden27}, the complexity
% analysis above remains valid when replacing $\totaldegree$ by this bound.
%
% sparsity for the BKK bound
% \cite{Bernstein75, Kouchnirenko76, Khovanskii78}
% \cite{Waerden27} for multi-homogeneous bound

We have implemented this algorithm and observed practical improvements over the
state-of-the-art algorithms on several examples. Notably, on examples where the
structure can be leveraged, e.g.\ where $\totaldegree$ is much lower than
$d(d-1)^{n-1}$, we observe competitive timings even compared to state-of-the-art
numerical solvers. See \Cref{sec:exp} for more details.

The results used to show the correctness of our algorithm are, due to its
nature, inherently geometric, and are interesting on their own. For this reason,
we have grouped all such results in \Cref{sec:correct}.

\subsection{Organisation of the paper}

We start by fixing notation and recalling basic notions in \Cref{sec:prelim}. We
then describe the algorithm and its subroutines in \Cref{sec:algo}.
\Cref{sec:correct} provides proofs of correctness of these subroutines and of
the main algorithm, \nameref{algo:smooth}. We analyse the probability of success
of the algorithm in \Cref{sec:proba}, and then analyse its bit cost in
\Cref{sec:comp}. Finally, \Cref{sec:exp} is devoted to practical
implementations.

\section{Notation and preliminaries}\label{sec:prelim}

\subsection{Basic notions}

We begin by recalling some basic algebraic and geometric notions which we use in
this paper. Those are standard, and we refer interested readers to
\cite{mumford,eisenbud,shafarevich,cox} for more details.

\paragraph*{Algebraic sets.}
Let $\field$ be a field and $\overline{\field}$ denote its algebraic closure. A
$\field$\emph{-algebraic set} $W \subset \overline{\field}^n$ is the set of
common solutions, in $\overline{\field}^n$, to $n$-variate polynomial equations,
where those polynomials have coefficients in $\field$. When the base field is
clear from context, we simply refer to those sets as \emph{algebraic sets}.
If explicit polynomials $g_1, \dots, g_s \in \field[x_1, \dots, x_n]$ are given,
we denote by $\V(g_1, \dots, g_s) \subset \overline{\field}^n$ the algebraic set
defined by $g_1 = \dots = g_s = 0$. Conversely, the \emph{ideal}, in
$\field[x_1, \dots, x_n]$, associated to a set $W
\subset \overline{\field}^n$ is the set of $n$-variate polynomials that vanish
on all points of $W$, and is denoted $\I(W)$.

An algebraic set $W \subset \overline{\field}^n$ is said to be
\emph{irreducible} when $W = W_1 \cup W_2$, with $W_1,W_2$ algebraic sets,
implies $W=W_1$ or $W=W_2$. Any algebraic set $W$ can be decomposed into a
finite union of irreducible algebraic sets, called the \emph{irreducible
components} of $W$, which are uniquely defined up to order. An algebraic set is
said to be a \textit{hypersurface} when it is defined by a single polynomial.

For an algebraic set $W$, its \emph{dimension}, denoted $\dim(W)$, is the Krull
dimension of the coordinate ring of $W$. By convention, we set $\dim(\emptyset)
\coloneqq -1$; zero-dimensional algebraic sets are precisely non-empty finite
algebraic sets. An algebraic set $W = \V(g_1, \dots, g_s)$ is said to be
$d$\emph{-equidimensional} when all its irreducible components have the same
dimension $d$.

\paragraph*{Zariski topology and genericity.}
The \emph{Zariski topology} on $\field^n$ is a topology where the closed sets
are of the form $\V(F) \cap \field^n$, where $F$ is any set of $n$-variate
polynomials with coefficients in $\field^n$. The \emph{Zariski closure} of
some set $X \subset \field^n$, denoted $\overline{X}^\calZ$, is the smallest
Zariski closed set containing $X$. Note that this is typically different from
the Euclidean closure $\overline{X}$.

A property $P$ depending on parameters $\bmp \in \field^n$ is said to be
\emph{generic} if there exists a non-empty Zariski open set $U \subseteq
\field^n$ such that $P(\bmp)$ is true for all $\bmp \in U$. 

\paragraph*{Quantitative aspects.}
The \textit{degree} of a zero-dimensional algebraic set is its cardinality.
For an irreducible algebraic set $W$, its \emph{degree} (in the sense of
\cite{Heintz83}), denoted $\deg(W)$, is the finite number of points contained in
the intersection of $W$ with $\dim(W)$ generic hyperplanes. For an arbitrary
algebraic set, its degree is the sum of the degrees of its irreducible
components of highest dimension $\dim(W)$. If the algebraic set is a
hypersurface, its degree is simply the degree of its defining polynomial, and we
formally define $\deg(\emptyset) = 1$.

For a $d$-equidimensional algebraic set $W \subset \compfield^n$ defined over
$\ratfield$, its \emph{height} (in the sense of \cite{Phil95,KPS01,DS04,DKS12})
is a measure of the bit complexity of its representation. Formally, it is
defined as
\[\htt(W) \coloneqq \sum_{p \text{ prime}} \ell_p(C_W) + m(C_W, d+1, n+1) +
(d+1)\deg(W)\sum_{i=1}^n \frac{1}{2i},\]
where $C_W$ denotes any Chow form of $W$ defined over $\ratfield$ \cite[Section
1]{KPS01}, $m(f,r,n)$ denotes the $(r,n)$-Mahler measure of $f$ over the complex
sphere of dimension $n$ \cite[Section 1]{KPS01}, and $\ell_p$ denotes the
$p$-adic height \cite[Section 1]{KPS01}. As a formal definition of all these
quantities would be rather technical, and not the chief topic of this paper, we
refer interested readers to e.g.\ \cite[Section 3]{DKS12} for a comprehensive
definition. 

\paragraph*{Smoothness, projections and critical points.}
% Given $g_1, \dots, g_s \in \field[x_1, \dots, x_n]$, we define the
% \emph{Jacobian} of $g_1, \dots, g_s$ by
% \[\Jac(g_1, \dots, g_s) \coloneqq \begin{bmatrix}
% 	\frac{\partial g_1}{\partial x_1} & \dots & \frac{\partial g_1}
%     {\partial x_n} \\
% 	\vdots & \ddots & \vdots \\
% 	\frac{\partial g_s}{\partial x_1} & \dots & \frac{\partial g_s}
%     {\partial x_n} \\
% \end{bmatrix}.\]
% We write $\Jac_{\bmy}(g_1, \dots, g_s)$ for the evaluation of the Jacobian at a
% point $\bmy \in \field^n$.
Assume that $W$ is a $d$-equidimen\-sional algebraic set. The \emph{Zariski
tangent space} to $W$ at $\bmx \in W$, denoted ${\T}_{\bmx}W \subset
\overline{\field}^n$, is the vector space defined by
\[\bmv \in {\T}_{\bmx}W \iff \grad(g)(\bmx) \cdot \bmv = 0 \text{ for all } g
\in \I(W).\]
If, as a vector space, $\dim({\T}_{\bmx}W) = d$, then $\bmx$ is said to be a
\emph{regular} or \emph{smooth} point of $W$. Otherwise, $\bmx$ is said to
be a \emph{singular} point of $W$. The algebraic set $W$ is said to be
\emph{smooth} if it does not have any singular point.
Furthermore, if $G = \{g_1, \dots, g_s\}$ generates $\I(W)$, then at any regular
point $\bmx \in W$, $\Jac_{\bmx}(G)$ has rank $n - d$ and its right kernel is
${\T}_{\bmx}W$, where $\Jac_{\bmx}(G)$ is the Jacobian matrix of $G$ evaluated
at $\bmx$.

Let $W \subset \compfield^n$ be a $d$-equidimensional algebraic set, and let
$\phi: \compfield^n \to \compfield^m$ be a polynomial map. A point $\bmx$ of $W$
is said to be a \emph{critical point} of the restriction of $\phi$ to $W$ if it
is a regular point of $W$ such that $\mathbin{{\der}\phi}\left(\T_{\bmx}W\right)
\neq \compfield^m$, where $\mathbin{{\der}\phi}$ denotes the
\emph{differential map} of the restriction of $\phi$ to $W$, defined as
$\mathbin{{\der}\phi}:\T_{\bmx}W \to \compfield^m, \bmu \mapsto
\Jac_{\bmx}(\phi)\cdot\bmu$.

We denote by $\crit(\phi, W)$ the set of critical points of $\phi$ on $W$. A
\emph{critical value} of $\phi$ on $W$ is a value $\bmy \in \compfield^m$ such
that $\exists \bmx \in \crit(\phi, W)$ with $\phi(\bmx) = \bmy$. As the above
definition of critical points does not allow us to directly compute them, we
recall the following lemma:

\noindent \cite[Lemma A.2]{SS17} \emph{Let $W$ be a $d$-equidimensional
algebraic set, and suppose that $G = \{g_1, \dots, g_s\}$ generates $\I(W)$. Let
$\phi: \compfield^n \to \compfield^m$ be a polynomial map. Then}
\[\crit(\phi, W) = \left\{\bmx \in W \ \middle| \
\rank\left(\Jac_{\bmx}(G)\right) = n-d, \quad \rank
\begin{bmatrix}
    \Jac_{\bmx}(G)\\
    \Jac_{\bmx}(\phi)
\end{bmatrix}
< n - d + m\right\}.\]

Throughout this paper, we chiefly focus on critical points of projection maps
onto some coordinates, which we denote as follows, for $1 \leq i \leq n$:
\[\pi_i : \compfield^n \to \compfield^i, \quad \pi_i(x_1, \dots, x_n) = (x_1,
\dots, x_i).\]

\paragraph*{Change of variables.}
For a matrix $\changeofvarmatrix \in \field^{n \times n}$ and a polynomial $g
\in \field[x_1, \dots, x_n]$, we define
\[g^\changeofvarmatrix \coloneqq g(\changeofvarmatrix \bmX) \in \field[x_1,
\dots, x_n],\]
where $\bmX$ denotes the column vector $(x_1, \dots, x_n)^{\T}$, where the $\T$
superscript denotes transpose. Similarly, if $G = \{g_1, \dots, g_s\}$, we
define $G^\changeofvarmatrix \coloneqq \big(g_1^\changeofvarmatrix, \dots,
g_s^\changeofvarmatrix\big)$. Finally, for an algebraic set $W \subset
\overline{\field}^n$ and $\changeofvarmatrix$ invertible, we define
\[W^\changeofvarmatrix \coloneqq \left\{\changeofvarmatrix^{-1}\bmx \ \middle| \
\bmx \in W \right\}.\]
Note that this notation is consistent, in the sense that
$\V(G^\changeofvarmatrix) = \V(G)^\changeofvarmatrix$.

Part of the randomness of the algorithm comes from choosing a ``suitable''
$\changeofvarmatrix$ to change variables before performing computations. We show
that part of this ``suitability'' arises from the very structure of
$\changeofvarmatrix$, and as such, we introduce notation to denote the parts of
$\changeofvarmatrix$ of interest. We denote by $a_{i,j}$ the entries of
$\changeofvarmatrix$. For $0 \leq k \leq n$, we define the matrices $\upMat_k$
and $\lowMat_k$ as the right blocks of $\changeofvarmatrix$ of sizes $k \times
(n-k)$ and $(n-k) \times (n-k)$ respectively:
\[\changeofvarmatrix = \begin{bmatrix} *&\upMat_k \\ *&\lowMat_k\end{bmatrix}.\]
Finally, we denote by $\lowEnt_{i,j,k}$ the $(i,j)$-th entry of
$\lowMat_k^{-1}$, when it exists. In particular, $\lowEnt_{i,j,0}$ denotes the
$(i,j)$-th entry of $\changeofvarmatrix^{-1}$, when it exists.

We also sometimes need to consider the entries of the change-of-variables matrix
as indeterminates. When this is the case, we refer to those variables
as $\frakA$, and a polynomial in those variables with coefficients in
$\field$ as an element of $\field[\frakA]$.

\subsection{Notation and assumptions}

Throughout this paper, we denote by $f \in \ratfield[x_1, \dots, x_n]$ the input
polynomial defining the semi-algebraic set $S \coloneqq \left\{\bmx \in
\mathbb{R}^n \ | \ f(\bmx) \neq 0\right\}$. We denote by $V \coloneqq
\V(f) \subset \compfield^n$ the hypersurface defined by the vanishing of $f$.

The algorithm requires the use of an arbitrary point of $\ratfield^{n-1}$, which
is denoted $\genericcoordsigma$ in the remainder of this paper. We sometimes
need to consider the coordinates of $\genericcoordsigma$ as indeterminates. When
this is the case, we denote them $\frakS$.

There are four key assumptions that frequently appear in our statements. To
simplify notation, we define and name them here as follows:
\begin{itemize}
    \item \OGassumS: The hypersurface $V$ is smooth, and $f$ is squarefree.
    \item \OGassumHa: We have $\changeofvarmatrix \in \GL_n(\compfield)
    \setminus \calZ$, where $\calZ$ is the proper Zariski closed set of
    \cite[Theorem 1]{SS03}. This assumption is here to guarantee that
    the computed critical points are in finite number.
    \item \OGassumHb: The matrices $\lowMat_k$ are invertible for all $1\leq k
    \leq n$. This assumption is here to guarantee that the polynomials
    determining the critical points we compute are well-defined.
    \item \OGassumB: We have $\genericcoordsigma \in \compfield^{n-1} \setminus
    \calV$, where $\calV$ is the proper Zariski closed set of \cite[Theorem
    2.2]{EGS20}. This assumption is here to guarantee that the smoothness of the
    boundary of the semi-algebraic set $S$ is preserved under instantiation of
    variables to these values.
\end{itemize}

Assumption \assumS\ is the central assumption of this paper. We require $f$ to
be squarefree for the ideal generated by $f$ to be radical, in order to apply
\cite[Lemma A.2]{SS17} for critical points computations. As stated before, this
is a mild assumption, as both $f$ and its squarefree part define the same
semi-algebraic set $S$.

Assumptions \assumHa, \assumHb\ and \assumB, on the other hand, explicitly state
the sufficiently generic conditions that $\changeofvarmatrix$ and
$\genericcoordsigma$ need to satisfy in order to be considered ``suitable'' for
our algorithm. We show that these conditions suffice for the smooth-case
algorithm in \Cref{sec:correct}, as well as analyse the probability that
randomly chosen $\genericcoordsigma$ and $\changeofvarmatrix$ with bounded
entries satisfy those conditions, in \Cref{sec:proba}.

\section{Algorithm}\label{sec:algo}

We start with a brief overview of the main algorithm. It takes as inputs a
straight-line program $\GamSLP$ evaluating $f \in \ratfield[x_1, \dots, x_n]$
satisfying \assumS, and a real number $\epsproba$ between $0$ and $1$, which
allows to control the probability of success of the algorithm. The parameter
$\epsproba$ comes into play when choosing random integers for the change of
coordinates and the specialisation values: they are chosen uniformly at random
within a set whose size is a multiple of $\epsproba^{-1}$. Therefore, the
smaller $\epsproba$ is, the larger the bitsize of the quantities involved, but
the lower the chance of failure.
The algorithm is based on four key steps:
\begin{enumerate}
    \item\label{step:1} Project $V \cap \realfield^n$ on a randomly chosen line,
    such that the critical points of this projection are finite, and such that
    the image of $V \cap \realfield^n$ under this projection is closed.
    \begin{figure}[H]
    \centering
    \includegraphics{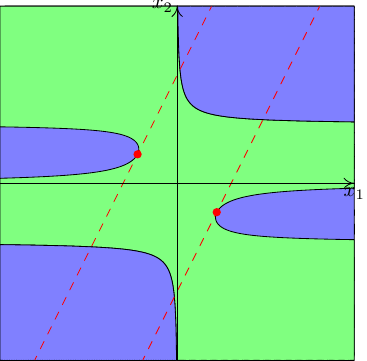}
    \caption{Illustration of the first step of algorithm applied to $f = 
    4x_1(x_2^3-x_2)-1$}
    \label{fig:step1}
    \end{figure}
    \item\label{step:2} Compute a rational parametrisation of these critical
    points, and obtain ``sufficiently precise'' isolation intervals for these
    points, in a transverse direction to $V \cap \realfield^n$.
    \begin{figure}[H]
    \centering
    \includegraphics{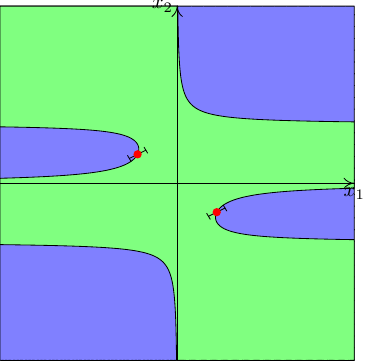}
    \caption{Illustration of the second step of algorithm applied to $f = 
    4x_1(x_2^3-x_2)-1$}
    \label{fig:step2}
    \end{figure}
    \item\label{step:3} For each endpoint of these intervals, compute a
    sufficiently small $n$-dimensional box around them, such that $f$ never
    vanishes at any point in the box. Pick one rational point in each such box.
    \begin{figure}[H]
    \centering
    \includegraphics{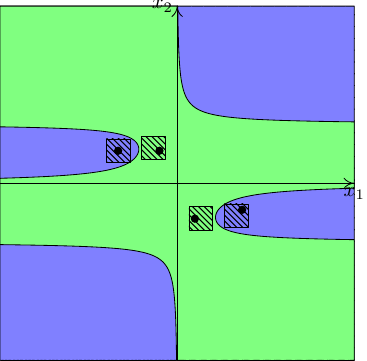}
    \caption{Illustration of the third step of algorithm applied to $f = 
    4x_1(x_2^3-x_2)-1$ (Exaggerated box size for visual clarity)}
    \label{fig:step3}
    \end{figure}
    \item\label{step:4} Instantiate $x_1$ to a randomly chosen value, say
    $\sigma_1$, in $f$ and perform the above steps with the newly obtained
    polynomial as
    input.
\end{enumerate}

We now introduce the subroutines we use in the main algorithm.

\noindent 
\(\bullet\) \ttSLPPolar\ takes as inputs the straight-line program $\GamSLP$
evaluating $f$, a matrix $\changeofvarmatrix\in\GL_n(\ratfield)$ satisfying
\assumHb, an integer $1\leq k \leq n$ and the specialisation point
$\genericcoordsigma \in \ratfield^{n-1}$. It outputs a straight-line program
evaluating the polynomials $\calP_k$, defined as
\begin{align*}
&\calP_k \coloneqq \Bigg(\sum_{i=1}^n \lowEnt_{1,i,0}x_i - \sigma_1, \dots, 
\sum_{i=1}^n 
\lowEnt_{k-1,i,0}x_i-\sigma_{k-1}, f, \\
&\indent \indent \frac{\partial f}{\partial x_{k+1}} + 
\sum_{i=1}^k \bigg(\sum_{j=1}^{n-k}a_{i,(j+k)}\lowEnt_{j,1,k}\bigg)
\frac{\partial f}{\partial x_i},\dots,\frac{\partial f}{\partial x_n} + 
\sum_{i=1}^k \bigg(\sum_{j=1}^{n-k}a_{i,(j+k)}\lowEnt_{j,n-k,k}\bigg)
\frac{\partial f}{\partial x_i}\Bigg)
\end{align*}
where $a_{i,j}$ denotes the entries of $\changeofvarmatrix$ and
$\lowEnt_{i,j,k}$ of $\lowMat_k^{-1}$.
%\cite[Theorem 1]{BS83} 

\noindent
\(\bullet\) \ttRatParam\ takes as inputs a straight-line program evaluating a
zero-dimensional polynomial system and $k\in \Nintegers \setminus \{0\}$. When
$k=1$, it outputs a zero-dimensional rational parametrisation of its regular
solutions with probability of success greater than or equal to $1 - (11/32)$.
Otherwise, it either outputs a parametrisation of a proper subset of the
solutions, or \texttt{fail}. When $k>1$, it repeats the procedure $k$ times, and
hence has a probability of success greater than or equal to $1 - (11/32)^k$,
since bad outputs can be ruled out by their degree. Such a subroutine can be
found in \cite[Algorithm 2]{SS18}.

We now state our first subroutine, \nameref{subr:critpts}, which returns a
rational para\-metrisation of the critical points of the projection.

\begin{subroutine}[H]
\subrCaptionLabel{CriticalPoints}{\GamSLP, \changeofvarmatrix, k,
\genericcoordsigma, r}{critpts}
\begin{algorithmic}[1]
    \Require $\GamSLP$, $\changeofvarmatrix$, $k$, $\genericcoordsigma$, $r$,
    where $\GamSLP$ is a straight-line program evaluating $f \in \ratfield[x_1,
    \dots, x_n]$, $\changeofvarmatrix \in \ratfield^{n \times n}$, $1 \leq k
    \leq n$, $\genericcoordsigma \in \ratfield^{n-1}$ and $r \in
    \Nintegers \setminus \{0\}$.
    \Assume $f$, $\changeofvarmatrix$ and $\genericcoordsigma$ satisfy \assumS,
    \assumHa, \assumHb\ and \assumB.
    \Ensure A zero-dimensional rational parametrisation $\scrQ = (w, \bmv,\nu)$.
    \State\label{lin:subr1:SLP} $\LamSLP \gets$ \ttSLPPolar$(\GamSLP,
    \changeofvarmatrix, k,\genericcoordsigma)$;
    \State\label{lin:subr1:param} $\scrQ \gets$ \ttRatParam$\left(\LamSLP,
    r\right)$;
    \State\label{lin:subr1:out} \Return $\scrQ$
\end{algorithmic}
\end{subroutine}

\noindent
\(\bullet\) \ttExtensionRealRoots\ takes as inputs a univariate polynomial $p
\in \ratfield[\varT]$ and a bivariate polynomial $q \in \ratfield[\varT,\varU]$.
Denoting by $t_1 < \dots < t_j$ the real algebraic roots of $p$ and by
$u_{i,1} < \dots < u_{i,\ell_i}$ the real roots of $q(t_i,\varU)$, and assuming
$q(t_i,\varU)$ is squarefree for all $i$, it outputs $j$ lists consisting, for
any $1 \leq i \leq j$, of $\ell_i$ disjoint intervals with rational endpoints
$[q_{i,k}^-, q_{i,k}^+]$ such that $u_{i,k} \in [q_{i,k}^-, q_{i,k}^+]$.
% \cite[Theorem 16]{STRZEBONSKI19}

We now state our second subroutine, \nameref{subr:isol}, which computes two
rational para\-metrisations of the endpoints of the isolation intervals of the
critical points in a direction transverse to $V \cap \realfield^n$. These
endpoints are the sample points of the semi-algebraic set $S$.

\begin{subroutine}[H]
\subrCaptionLabel{Isolation}{\GamSLP, \bma, \scrQ}{isol}
\begin{algorithmic}[1]
    \Require $\GamSLP$, $\bma$, $\scrQ$, where $\GamSLP$ is a straight-line
    program evaluating $f \in \ratfield[x_1, \dots, x_n]$, $\bma \in
    \ratfield^{n}$, and $\scrQ = (w, \bmv, \nu)$ is a zero-dimensional rational
    parametrisation, where $\bmv = (v_1, \dots, v_n)$.
    \Assume $f$ satisfies \assumS, $w \not\equiv 1$, and for any $\bxi$
    parametrised by $\scrQ$, $\bxi \in V$ and the line parametrised by $\bxi +
    \bma\varU$ does not intersect $V$ tangentially at $\bxi$.
    \Ensure Two zero-dimensional rational parametrisations $\scrP^-$ and 
    $\scrP^+$.
    \State\label{lin:subr2:degree} $d \gets \deg(f)$;
    \State\label{lin:subr2:interVl} $\scrI \gets$ \ttExtensionRealRoots$\left(
    w(\varT), \left(w'(\varT)\right)^df\left(\frac{1}{w'(\varT)}\bmv(\varT) +
    \varU\bma\right)\right)$;
    \State\label{lin:subr2:lambda} $\lambda \gets \min\left\{1, \min\{|e|: e \in
    \{\text{endpoints of each interval in }\scrI\} \wedge e \neq 0\}\right\}/2$;
    \State\label{lin:subr2:P-} $\scrP^- \gets (w, \bmv -\lambda w'\bma,\nu)$;
    \State\label{lin:subr2:P+} $\scrP^+ \gets (w, \bmv +\lambda w'\bma,\nu)$;
    \State\label{lin:subr2:out} \Return $\scrP^-, \scrP^+$
\end{algorithmic}
\end{subroutine}

\noindent
\(\bullet\) \ttParamMinValue\ \& \ttParamMaxValue\ both take as inputs a
zero-dimensional rational parametrisation $\scrQ$ and a polynomial $f$ with
rational coefficients. They respectively output a positive lower and upper bound
on the non-zero absolute values the polynomial $f$ takes at any real solution to
the parametrisation $\scrQ$.
% Is not immediately present in literature so needs to be explicited.

\noindent
\(\bullet\) \ttApproximationList\ takes as inputs a zero-dimensional rational
parametrisation $\scrQ$ and an integer $r$, and outputs rational approximations
of all real solutions to $\scrQ$, whose coordinates are accurate to precision
$2^{-r}$. 
%~\cite[Theorem 47]{melczersalvy}

We now state our third and last subroutine, \nameref{subr:approx}, which
computes rational approximations of the previously computed sample points, which
are precise enough to lie in the same connected components as their exact
algebraic counterpart.

\begin{subroutine}[H]
\subrCaptionLabel{Approximation}{\GamSLP, \scrQ}{approx}
\begin{algorithmic}[1]
    \Require $\GamSLP$, $\scrQ$, where $\GamSLP$ is a straight-line program
    evaluating $f \in \ratfield[x_1, \dots, x_n]$, and $\scrQ = (w, \bmv, \nu)$
    is a zero-dimensional rational parametrisation, where $\bmv = (v_1, \dots,
    v_n)$.
    \Assume $f$ is non-zero on every point parametrised by $\scrQ$.
    \Ensure A finite subset $\scrA$ of $\ratfield^n$.
    \State\label{lin:subr3:degree} $d \gets \deg(f)$;
    \State\label{lin:subr3:maxcoef} $\scrc \gets \norm{f}_\infty$;
    \State\label{lin:subr3:minval} $\scrm \gets$ \ttParamMinValue
    $\left(\scrQ,f\right)$;
    \State\label{lin:subr3:maxval} $\scrM \gets \max_{1 \leq i \leq n}$
    \ttParamMaxValue$\left(\scrQ, x_i\right) + 1$;
    \State\label{lin:subr3:dist} $\scrd \gets \min\left\{1,\scrm \left(
    \binom{n+d}{d}^2\scrM^{d}\scrc \right)^{-1}\right\}$;
    \State\label{lin:subr3:approx} $\scrA \gets$ \ttApproximationList$\left(
    \scrQ, \left\lceil-\log(\scrd)\right\rceil\right)$;
    \State\label{lin:subr3:out} \Return $\scrA$
\end{algorithmic}
\end{subroutine}

\noindent We can now describe the main algorithm, \Cref{algo:smooth}. The
algorithm essentially reduces to three previously introduced subroutines.

\begin{algorithm}[H]
\algoCaptionLabel{Point\-Per\-Connected\-Comp\-onent\-Smooth}{\GamSLP,\epsproba}{smooth}
\begin{algorithmic}[1]
    \Require $\GamSLP, \epsproba$ where $\GamSLP$ is a straight-line program
    evaluating $f \in \ratfield[x_1, \dots, x_n]$, and $0 < \epsproba < 1$.
    \Assume $f$ satisfies \assumS.
    \Ensure A finite subset of $\ratfield^n$.
    \State\label{lin:gen:degree} $d \gets \deg(f)$;
    \State\label{lin:gen:A} Build $\setD \coloneqq \left\{1, \dots, \left\lceil
    3\epsproba^{-1}\left(5n^3(2d)^{2n}+\frac{n^2-n}{2}\right)\right\rceil
    \right\}$ and draw $\changeofvarmatrix \in \setD^{n \times n}$ at random;
    \State\label{lin:gen:sigma} Build $\setE \coloneqq\left\{1, \dots, \left
    \lceil 3\epsproba^{-1}nd^{2n}\right\rceil \right\}$ and draw
    $\genericcoordsigma \in \setE^{n-1}$ at random;
    \For{$k \in \{1,\dots,n\}$}\label{lin:gen:mainfor}
    \State\label{lin:gen:crit} $\scrQ_k =(w_k, \bmv_k, \nu_k) \gets$
    \procCrit$\left(\GamSLP, \changeofvarmatrix, k, \genericcoordsigma,
    \left\lceil\log(3n\epsproba^{-1})\right\rceil\right)$;
    \If{$w_k \not\equiv 1$}\label{lin:gen:iftrivial}
    \State\label{lin:gen:Pparam} $\scrP_k^-, \scrP_k^+ \gets$
    \procIsol$(\GamSLP, (a_{1,k}, \dots, a_{n,k}), \scrQ_k)$;
    \State\label{lin:gen:approx-} $\scrA_k^- \gets$ \procApprox
    $(\GamSLP,\scrP_k^-)$;
    \State\label{lin:gen:approx+} $\scrA_k^+ \gets$ \procApprox
    $(\GamSLP,\scrP_k^+)$;
    \State\label{lin:gen:Ak} $\scrA_k \gets \scrA_k^- \cup \scrA_k^+$;
    \Else \State\label{lin:gen:empty} $\scrA_k \gets \emptyset$;

    \EndIf

    \EndFor
    
    \State\label{lin:gen:output} \Return $\bigcup_{k=1}^n \scrA_k \cup
    \{\changeofvarmatrix(\sigma_1,\dots,\sigma_{n-1},0)^{\T}\}$
\end{algorithmic}
\end{algorithm}

In the illustrated example, the first three steps correspond to \Cref{%
lin:gen:crit,lin:gen:iftrivial,lin:gen:Pparam,lin:gen:approx-,lin:gen:approx+},
and the last step corresponds to the \texttt{for} loop of
\Cref{lin:gen:mainfor}.

\section{Correctness}\label{sec:correct}

\subsection{Correctness of \texorpdfstring{\nameref{subr:critpts}}{}}
\label{subsec:critpts}

In this subsection, we show that, under assumptions \assumS, \assumHa, \assumHb\
and \assumB, \nameref{subr:critpts} correctly outputs a zero-dimensional
rational parametrisation $\scrQ$ which encodes the sought-for critical points.

We begin by some intuitive geometric results on the boundary of a connected
component of the semi-algebraic set $S$.

\begin{lemma}\label{lem:connectedboundary}
    For any non-isolated point $\bmy \in V \cap \realfield^n$, there exists a
    neighbourhood $\scrN_1$ of $\bmy$ such that $\scrN_1 \cap V \cap
    \realfield^n$ is connected.
\end{lemma}
\begin{proof}
    Let us fix a non-isolated $\bmy \in V \cap \realfield^n$. Then, by the local
    conic structure theorem \cite[Theorem 9.3.6]{BCR}, there exists $e \in
    \realfield$, $e > 0$, and a semi-algebraic homeomorphism $\phi$ such that:
    \begin{itemize}
        \item $\phi: \overline{B}_n(\bmy,e) \to \overline{B}_n(\bmy,e)$, where
        $\overline{B}_n(\bmy,e)$ denotes the closed $n$-dimensional ball of
        radius $e$ centred at $\bmy$,
        \item $\norm{\phi(\bmx) - \bmy} = \norm{\bmx - \bmy}$ for every $\bmx
        \in \overline{B}_n(\bmy,e)$,
        \item $\phi$ restricted to $S^{n-1}(\bmy,e)$ is the identity map, where
        $S^{n-1}(\bmy,e)$ denotes the boundary of $\overline{B}_n(\bmy,e)$,
        \item and $\phi^{-1}(V \cap \realfield^n \cap \overline{B}_n(\bmy,e))$
        is the cone with vertex $\bmy$ and basis $V \cap \realfield^n \cap
        S^n(\bmy,e)$.
    \end{itemize}
    In particular, since $\phi^{-1}(V \cap \realfield^n \cap
    \overline{B}_n(\bmy,e))$ is a cone, it is connected. Moreover, since
    homeomorphisms preserve connectivity, it follows that $\phi(\phi^{-1}(V \cap
    \realfield^n \cap \overline{B}_n(\bmy,e))) = V \cap \realfield^n \cap
    \overline{B}_n(\bmy,e)$ is also connected. Letting $\scrN_1 \coloneqq
    \overline{B}_n(\bmy,e)$ finishes the proof.
\end{proof}

\begin{lemma}\label{lem:connectedsemialg}
    Suppose that $f$ satisfies \emph{\assumS}, and let $S^+ \coloneqq
    \left\{\bmx \in \realfield^n \ | \ f(\bmx) > 0\right\}$. Then, for any $\bmy
    = (y_1, \dots, y_n) \in V \cap \realfield^n$, there exists a neighbourhood
    $\scrN_2$ of $\bmy$ such that $\scrN_2 \cap S^+$ is connected.
\end{lemma}
\begin{proof}
    Since $f$ satisfies \assumS, the gradient of $f$ at $\bmy$ is non-zero, and
    hence there exists $1 \leq i \leq n$ such that its $i$-th coordinate is
    non-zero. Without loss of generality, we assume for simplicity that $i=n$.

    Then, by the semi-algebraic Implicit Function Theorem \cite[Corollary
    2.9.8]{BCR}, there exist open semi-algebraic neighbourhoods $U_1 \subseteq
    \realfield^{n-1}$ and $U_2 \subseteq \realfield$ of $(y_1, \dots, y_{n-1})$
    and $y_n$ respectively, and a map $\phi: U_1 \to U_2$ such that $\phi(y_1,
    \dots, y_{n-1}) = y_n$ and $f(\bmx) = 0 \iff x_n = \phi(x_1, \dots,
    x_{n-1})$ for every $\bmx \in U_1 \times U_2$. Let us now define
    \[A^+ \coloneqq \left\{\bmx \in U_1 \times U_2 \ | \ x_n > \phi(x_1, \dots, 
    x_{n-1})\right\}.\]
    Let $\scrN_1$ be the neighbourhood of \Cref{lem:connectedboundary}, and let
    $\balpha,\betta \in A^+ \cap \scrN_1$ be distinct. If $\balpha$ and $\betta$
    can be path connected, then $A^+ \cap \scrN_1$ is connected (since we are
    working in $\realfield^n$), which by the definition of $\phi$ implies that
    $(U_1 \times U_2) \cap \scrN_1 \cap S^+$ is connected, and then defining
    $\scrN_2 \coloneqq (U_1 \times U_2) \cap \scrN_1$ finishes the proof. It
    therefore suffices to show that $\balpha$ and $\betta$ are path connected in
    $A^+ \cap \scrN_1$.

    By definition of $\scrN_1$ and by \Cref{lem:connectedboundary}, $V \cap
    \realfield^n \cap \scrN_1$ is connected. Hence, there exists a path
    $(\gamma_1(t), \dots, \gamma_n(t))$, for $t \in [0,1]$, from $(\alpha_1,
    \dots, \alpha_{n-1}, \phi(\alpha_1, \dots, \alpha_{n-1}))$ to $(\beta_1,
    \dots, \beta_{n-1}, \phi(\beta_1, \dots, \beta_{n-1}))$. We then consider
    the path $(\gamma_1(t), \dots, \gamma_{n-1}(t), \gamma_n(t)+\alpha_n t +
    \beta_n(1-t))$. It is a path from $\balpha$ to $\betta$, which lies in $A^+
    \cap \scrN_1$, since $\gamma_n(t)+\alpha_n t + \beta_n(1-t) >
    \phi(\gamma_1(t), \dots, \gamma_{n-1}(t))$. Hence, $\balpha$ and $\betta$
    are path connected in $A^+ \cap \scrN_1$, as required.
\end{proof}

\begin{lemma}\label{lem:boundary} Suppose that $f$ satisfies \emph{\assumS} and
let $C$ be a connected component of $S$. Then, each connected component of the
boundary $\overline{C} \setminus C$ of $C$ is a connected component of the
real algebraic set $V \cap \realfield^n$.
\end{lemma}
\begin{proof}
Without loss of generality, suppose that $f$ is positive on $C$, that is, $C$ is
a connected component of $S^+$. Let
$\bmx \in \overline{C}$, and suppose for a contradiction that
$f(\bmx) < 0$. Then, there exists an open ball $\calB$ around $\bmx$
such that $f$ is negative on every point of that ball, by continuity of $f$.
In particular, $\calB \cap C = \emptyset$, which contradicts the fact that
$\bmx$ belongs to the Euclidean closure of $C$. Therefore, $\bmx
\in \overline{C}$ implies that $f(\bmx) \geq 0$, and hence
$$\overline{C} \setminus C \subseteq \left\{\bmx \in \realfield^n \mid f(\bmx)
\geq 0 \wedge f(\bmx) \not> 0\right\} = V \cap \realfield^n.$$ 
Hence, a connected component of $\overline{C} \setminus C$ necessarily lies in
$V \cap \realfield^n$. Let \(B\subset V\cap\realfield^n\) be a connected
component of  \(\overline{C}\setminus C \). Then, there exists a connected
component \(A\) of \(V\cap \realfield^n\) which contains \(B\), since otherwise
\(B\) itself would not be connected. To prove the claim, it thus suffices to
show that $B=A$.

Let $\bmx \in B \subset A$ and $\bmy \in A$. Since $A$ is connected, there
exists a continuous path $\gamma:[0,1] \to A$ such that $\gamma(0) = \bmx$ and
$\gamma(1) = \bmy$. For any $t \in [0,1]$, by
\Cref{lem:connectedboundary,lem:connectedsemialg}, there exists a neighbourhood
$\scrN_t \coloneqq \scrN_{1,t} \cap \scrN_{2,t}$ around $\gamma(t)$ such that
$\scrN_t \cap V \cap \realfield^n$ and $\scrN_t \cap S^+$ are connected. We now
consider the sets \(\gamma([0, 1])\cap 
\mathscr{N}_t\). These sets are open (as $\scrN_t$'s are), non-empty (as they
contain $\gamma(t)$), and cover $\gamma([0,1])$, which is compact (as $[0,1]$
is). Hence, there exists a finite covering \(\cup_{t\in \{t_0 < \cdots <
t_k\}}\mathscr{N}_t \cap \gamma([0, 1])\) of \(\gamma([0, 1])\) with  \(t_0 =
0\) and \(t_k = 1\). For $1 \leq i \leq k$, pick $\bmz_i \in
\mathscr{N}_{t_{i-1}} \cap \mathscr{N}_{t_{i}} \cap \gamma([0, 1])$, and
formally define $\bmz_0$ as $\bmx$ and $\bmz_{k+1}$ as $\bmy$. If we show that
$\bmz_i \in B$ implies $\bmz_{i+1} \in B$ for all $i$, this implies, since
$\bmz_0 = \bmx \in B$, that $\bmz_{k+1} = \bmy$ lies in $B$, which proves the
claim. It therefore suffices to show that $\bmz_i \in B \implies \bmz_{i+1} \in
B$.

Let us suppose that $\bmz_i \in B$ for some $i$. Then, as both $\scrN_{t_i} \cap
V \cap \realfield^n$ and $\scrN_{t_i} \cap S^+$ are connected, it follows that
any point of $\scrN_{t_i} \cap \gamma([0,1])$ lies in the same connected
component of the boundary of the same connected component of $S^+$ as $\bmz_i$.
In particular, as $\bmz_i \in B$ by assumption, and as $\bmz_{i+1} \in
\scrN_{t_i}$ by construction, it follows that $\bmz_{i+1} \in B$, as required.
\end{proof}

We recall here some of the results from \cite{SS03}:
\cite[Theorem 1 and Proposition 4]{SS03} \textit{There exists a proper Zariski
closed set $\mathcal{Z}$ of $\GL_n(\compfield)$ such that, for any
$\changeofvarmatrix \in \GL_n(\compfield)\setminus\mathcal{Z}$ and $1 \leq k
\leq n$, $V^\changeofvarmatrix$ is in Noether position with respect to $\pi_k$
and, for any connected component \(C\) of  \(V^\changeofvarmatrix\cap
\realfield^n\), the boundary of $\pi_k(C)$ is contained in
$\pi_k(\critpoints(\pi_k, V^\changeofvarmatrix) \cap C)$.}

\begin{lemma}\label{lem:projclosed} 
    Suppose that $f$ satisfies \emph{\assumS}. Let $B$ be a connected component
    of $V \cap \realfield^n$, and let $\calZ \subset \compfield^{n \times n}$ be
    the proper Zariski closed set of \emph{\cite[Theorem 1]{SS03}}. Then, for
    any $\changeofvarmatrix \in \GL_n(\compfield)\setminus \mathcal{Z}$ and for
    any $1 \leq k \leq n$, $\pi_k(B^\changeofvarmatrix)$ is closed for the
    Euclidean topology, where $B^\changeofvarmatrix$ denotes the
    connected component of $V^\changeofvarmatrix \cap \realfield^n$
    corresponding to $B$ after application of the change of variables
    $\changeofvarmatrix$.
\end{lemma}
\begin{proof}
Since changing variables by $\changeofvarmatrix$ preserves connectivity, we
deduce that $B^\changeofvarmatrix$ is a connected component of
$V^\changeofvarmatrix \cap \realfield^n$, and hence, for any
$\changeofvarmatrix$ and $k$ as above, by \cite[Theorem 1 and Proposition
4]{SS03}, we have
\[\overline{\pi_k(B^\changeofvarmatrix)} \setminus \pi_k(B^\changeofvarmatrix)
\subseteq \pi_k(\critpoints(\pi_k, V^\changeofvarmatrix) \cap
B^\changeofvarmatrix) \subseteq \pi_k(B^\changeofvarmatrix)\] 
and hence $\overline{\pi_k(B^\changeofvarmatrix)} \subseteq
\pi_k(B^\changeofvarmatrix)$, which implies that the image of
$B^\changeofvarmatrix$ by $\pi_k$ is closed for the Euclidean topology. 
\end{proof}

Note that, with our notation, the required condition on $\changeofvarmatrix$
for \Cref{lem:projclosed} is precisely that $\changeofvarmatrix$ satisfies
\assumHa. For the convenience of the reader, we now recall the statement of
\cite[Theorem 2.2]{EGS20}:

\noindent \cite[Theorem 2.2]{EGS20}: \emph{Suppose that $f$ satisfies
\emph{\assumS}, and let $\calZ$ be as in \emph{\Cref{lem:projclosed}}. There
exists a non-zero polynomial $g \in \compfield[\frakS]$ of degree at most
$nd^{2n}$ such that, for any $\changeofvarmatrix \in \GL_n(\compfield) \setminus
\calZ$, $1 \leq k \leq n$, and $\genericcoordsigma \in \compfield^{n-1}
\setminus \calV$, where $\calV \coloneqq \V(g)$, the system of equations
\[x_1 - \sigma_1 = \dots= x_{k-1}-\sigma_{k-1}=f^\changeofvarmatrix=
\frac{\partial f^\changeofvarmatrix}{\partial x_{k+1}}= \dots=\frac{\partial
f^\changeofvarmatrix}{\partial x_n}=0\] has finitely many solutions, none of
which are singular}.

Note again that, with our notation, the required conditions on
$\changeofvarmatrix$ and $\genericcoordsigma$ for \cite[Theorem 2.2]{EGS20} to
be applicable are precisely that $\changeofvarmatrix$  satisfies \assumHa\ and
that $\genericcoordsigma$ satisfies \assumB.
By \cite[Lemma A.2]{SS17}, we have that the solutions to the systems of
\cite[Theorem 2.2]{EGS20} are precisely the critical points of
$V^\changeofvarmatrix$ under projection by $\pi_k$, for which the first $k-1$
coordinates have been instantiated to $\sigma_1, \dots, \sigma_{k-1}$, which
under our notation is the set $\crit(\pi_k, V^\changeofvarmatrix) \cap
\pi_{k-1}^{-1}(\sigma_1, \dots, \sigma_{k-1})$. We could therefore use this
system as the starting point of \Cref{algo:smooth}; however, we instead choose
to start with the related following system:

\begin{lemma}\label{lem:critsystem}
    Suppose that $\changeofvarmatrix$ is invertible and satisfies
    \emph{\assumHb}, and let $\calP_k$ denote the polynomials obtained in the
    output of \emph{\ttSLPPolar}. Then, for any $1 \leq k \leq n$, we have
    \[\bxi \in \V\Big(x_1 - \sigma_1,\dots, x_{k-1}-\sigma_{k-1},
    f^\changeofvarmatrix, \frac{\partial f^\changeofvarmatrix}{\partial x_{k+1}}
    ,\dots,\frac{\partial f^\changeofvarmatrix}{\partial x_n}\Big) \iff
    \changeofvarmatrix\bxi \in \V\left(\calP_k\right).\]
\end{lemma}
\begin{proof}
    By the (multivariate) chain rule, we have $\frac{\partial
    f^\changeofvarmatrix}{\partial x_j} = \sum_{i=1}^n a_{i,j}
    \left(\frac{\partial f}{\partial x_i}\right)^\changeofvarmatrix$, for all $1
    \leq j \leq n$. Therefore, by definition of the Jacobian matrix, we obtain
    $\Jac_{\changeofvarmatrix\bmx}(f) \changeofvarmatrix =
    \Jac_{\bmx} \left(f^\changeofvarmatrix\right)$. Let us fix $1 \leq k \leq
    n$. As $\changeofvarmatrix$ satisfies \assumHb, $\lowMat_k^{-1}$ exists, and
    hence
    \begin{align*}
        &\phantom{{}\iff{}} \bxi \in \V\left(\frac{\partial f^\changeofvarmatrix}
        {\partial x_{k+1}}, \dots, \frac{\partial f^\changeofvarmatrix}
        {\partial x_n}\right) \\
        &\iff \frac{\partial f^\changeofvarmatrix}
        {\partial x_{k+1}}(\bxi) = \dots = \frac{\partial f^\changeofvarmatrix}
        {\partial x_n}(\bxi) = 0 \\
        &\iff \begin{bmatrix} \frac{\partial f}{\partial x_1}(\changeofvarmatrix
        \bxi) & \cdots & \frac{\partial f}{\partial x_n}(\changeofvarmatrix\bxi)
        \end{bmatrix} \begin{bmatrix} \upMat_k \\ \lowMat_k \end{bmatrix} = 
        \begin{bmatrix} 0 & \cdots & 0 \end{bmatrix} \\
        &\iff \begin{bmatrix} \frac{\partial f}{\partial x_1}(\changeofvarmatrix
        \bxi) & \cdots & \frac{\partial f}{\partial x_n}(\changeofvarmatrix\bxi)
        \end{bmatrix} \begin{bmatrix} \upMat_k \\ \lowMat_k \end{bmatrix} 
        \lowMat_k^{-1} = \begin{bmatrix} 0 & \cdots & 0 \end{bmatrix} \\
        &\iff \begin{bmatrix} \frac{\partial f}{\partial x_1}(\changeofvarmatrix
        \bxi) & \cdots & \frac{\partial f}{\partial x_n}(\changeofvarmatrix\bxi)
        \end{bmatrix} \begin{bmatrix} \upMat_k\lowMat_k^{-1} \\ I_{n-k} 
        \end{bmatrix} = \begin{bmatrix} 0 & \cdots & 0\end{bmatrix} \\
        &\iff \begin{bmatrix} \frac{\partial f}{\partial x_{k+1}}(
        \changeofvarmatrix\bxi) & \cdots & \frac{\partial f}{\partial x_n}
        (\changeofvarmatrix\bxi) \end{bmatrix} + \begin{bmatrix} \frac{
        \partial f}{\partial x_1}(\changeofvarmatrix\bxi) & \cdots & \frac
        {\partial f}{\partial x_k}(\changeofvarmatrix\bxi) \end{bmatrix}
        \upMat_k\lowMat_k^{-1} = \begin{bmatrix} 0 & \cdots & 0\end{bmatrix} \\
        &\iff \changeofvarmatrix\bxi \in \V\Bigg(\frac{\partial f}
        {\partial x_{k+1}} + \sum_{i=1}^k \bigg(\sum_{j=1}^{n-k}a_{i,(j+k)}
        \lowEnt_{j,1,k}\bigg)\frac{\partial f}{\partial x_i}, \dots, \\
        &\indent\indent\indent\indent\indent \frac{\partial f}
        {\partial x_n} + \sum_{i=1}^k \bigg(\sum_{j=1}^{n-k}a_{i,(j+k)}
        \lowEnt_{j,n-k,k}\bigg)\frac{\partial f}{\partial x_i}\Bigg).
    \end{align*}
    On the other hand, as $\changeofvarmatrix$ is invertible, 
    $\lowMat_0^{-1}=\changeofvarmatrix^{-1}$ exists, and by substitution we have
    \begin{align*}
        &\phantom{{}\iff{}} \bxi \in \V\left(x_1 - \sigma_1, \dots, x_{k-1} - 
        \sigma_{k-1} \right) \\
        &\iff \changeofvarmatrix\bxi \in \V\left(\sum_{i=1}^n \lowEnt_{1,i,0}x_i- 
        \sigma_1, \dots, \sum_{i=1}^n\lowEnt_{k-1,i,0}x_i - \sigma_{k-1}\right).
    \end{align*}
    Since we also have $\bxi \in \V(f^\changeofvarmatrix) \iff 
    \changeofvarmatrix\bxi \in \V(f)$ by definition, we conclude by combining
    those results.
\end{proof}

In \Cref{lin:gen:crit} of \nameref{algo:smooth}, we use the system introduced in
\Cref{lem:critsystem} as the starting system to solve rather than the more
intuitive system of \cite[Theorem 2.2]{EGS20}. This is because the partial
derivatives of $f^\changeofvarmatrix$ are generically of maximal degree $d-1$
regardless of the structure of $f$ for any variable, whereas the partial
derivatives of $f$ could have much lower degree in some variables. This input
system therefore preserves part of the structure associated to $f$, and comes at
the cost of an extra condition \assumHb\ that $\changeofvarmatrix$ must satisfy.
Combining all these results yields:

\begin{lemma}\label{lem:subrcritcorrect}
    Suppose that $f$ satisfies \emph{\assumS}, that $\changeofvarmatrix$
    satisfies \emph{\assumHa} and \emph{\assumHb}, and that $\genericcoordsigma$
    satisfies \emph{\assumB}. Let us fix $1 \leq k \leq n$ and $r \in \Nintegers
    \setminus \{0\}$, and let $\GamSLP$ be a straight-line program evaluating
    $f$. If the call to \emph{\ttRatParam} is successful after $r$ attempts,
    then \emph{\nameref{subr:critpts}} applied to $\GamSLP$,
    $\changeofvarmatrix$, $k$, $\genericcoordsigma$ and $r$ correctly returns a
    zero-dimensional rational parametrisation of the set
    $\changeofvarmatrix\left(\crit(\pi_k, V^\changeofvarmatrix) \cap
    \pi_{k-1}^{-1}(\sigma_1, \dots, \sigma_{k-1})\right)$.
\end{lemma}
\begin{proof}
    As $f$ satisfies \assumS, $\changeofvarmatrix$ satisfies \assumHa\ and
    $\genericcoordsigma$ satisfies \assumB, we can apply \cite[Theorem
    2.2]{EGS20} to obtain that the set of solutions to the system 
    \[x_1 - \sigma_1 = \dots= x_{k-1}-\sigma_{k-1}=f^\changeofvarmatrix=
    \frac{\partial f^\changeofvarmatrix}{\partial x_{k+1}}= \dots=\frac{\partial
    f^\changeofvarmatrix}{\partial x_n}=0\]
    is finite and does not contain singular points. By \cite[Lemma A.2]{SS17},
    this solution set is precisely the set $\crit(\pi_k, V^\changeofvarmatrix)
    \cap \pi_{k-1}^{-1}(\sigma_1, \dots, \sigma_{k-1})$. Since
    $\changeofvarmatrix$ is invertible (as it satisfies \assumHa) and satisfies
    \assumHb, we can apply \Cref{lem:critsystem}, to deduce that the set of
    solutions to the system 
    \begin{align*}
    & \Bigg(\sum_{i=1}^n \lowEnt_{1,i,0}x_i - \sigma_1, \dots, \sum_{i=1}^n 
    \lowEnt_{k-1,i,0}x_i-\sigma_{k-1}, f, \frac{\partial f}{\partial x_{k+1}} + 
    \sum_{i=1}^k \bigg(\sum_{j=1}^{n-k}a_{i,(j+k)}\lowEnt_{j,1,k}\bigg)
    \frac{\partial f}{\partial x_i}, \\
    & \indent\indent\indent \dots,\frac{\partial f}{\partial x_n} + \sum_{i=1}^k 
    \bigg(\sum_{j=1}^{n-k}a_{i,(j+k)}\lowEnt_{j,n-k,k}\bigg)\frac{\partial f}
    {\partial x_i}\Bigg)
    \end{align*}
    is the set $\changeofvarmatrix\left(\crit(\pi_k, V^\changeofvarmatrix) \cap
    \pi_{k-1}^{-1}(\sigma_1, \dots, \sigma_{k-1})\right)$. In particular, it
    must also be finite and not contain singular points.
    
    Since $\changeofvarmatrix$ satisfies \assumHb, we can apply \ttSLPPolar\ on
    \Cref{lin:subr1:SLP}, which, given $\GamSLP$, $\changeofvarmatrix$, $k$ and
    $\genericcoordsigma$, computes a straight-line program $\LamSLP$ that
    evaluates this system. As it has finitely many solutions, we can apply
    \ttRatParam\ on \Cref{lin:subr1:param} to $\LamSLP$ and $r$, which, under
    assumption of success, computes a zero-dimensional rational parametrisation
    $\scrQ$ of the regular solutions to this system. Since $\changeofvarmatrix$
    satisfies \assumHa\ and $\genericcoordsigma$ satisfies \assumB, by
    \cite[Theorem 2.2]{EGS20}, its solutions are all regular, and $\scrQ$ is
    therefore a parametrisation of the entire set
    $\changeofvarmatrix\left(\crit(\pi_k, V^\changeofvarmatrix) \cap
    \pi_{k-1}^{-1}(\sigma_1, \dots, \sigma_{k-1})\right)$. 
    Since the output of \nameref{subr:critpts} is precisely
    $\scrQ$ (see \Cref{lin:subr1:out}), we conclude that \nameref{subr:critpts}
    correctly returns a zero-dimensional rational parametrisation of the set
    $\changeofvarmatrix\left(\crit(\pi_k, V^\changeofvarmatrix) \cap
    \pi_{k-1}^{-1}(\sigma_1, \dots, \sigma_{k-1})\right)$ when all the
    aforementioned assumptions are made, as required.
\end{proof}

\subsection{Correctness of \texorpdfstring{\nameref{subr:isol}}{}}
\label{subsec:isol}

Recall that $V = \V(f) \subset \compfield^n$ and that $S \coloneqq \left\{\bmx
\in \mathbb{R}^n \ | \ f(\bmx) \neq 0\right\}$. 
In this subsection, we show that, if $f$ satisfies \assumS, if the
zero-dimensional rational parametrisation $\scrQ$ parametrises a non-empty
finite set $Q$, and if $\bma$ denotes a direction vector not tangent to $V \cap
\realfield^n$ at any point of $Q$, then \nameref{subr:isol} correctly outputs
two zero-dimensional rational parametrisations, which encode points in each
connected component of $S$ whose boundary contains a point of $Q$.
We begin by a small lemma concerning the behaviour of the boundary of $S$.

\begin{lemma}\label{lem:twoconnectedcomp}
    Suppose that $f$ satisfies \emph{\assumS}. Then, any point in the boundary
    of $S$ lies in the boundary of at most two distinct connected components of
    $S$.
\end{lemma}
\begin{proof}
    Let $\bxi \in \overline{S} \setminus S$. Since $f$ satisfies \assumS, we can
    apply \Cref{lem:connectedsemialg} to conclude that there exists a
    neighbourhood $\mathscr{N}^+$ of $\bxi$ such that $\mathscr{N}^+ \cap S^+$
    is connected. Therefore, there can only be at most one connected component
    of $S$ such that $\bxi$ lies in its boundary and $f$ is positive on it. By
    applying the same argument to $-f$, we also deduce the existence of a
    neighbourhood $\mathscr{N}^-$ of $\bxi$ such that $\mathscr{N}^- \cap S^-$
    is connected, and hence there can only be at most one connected component of
    $S$ such that $\bxi$ lies in its boundary and $f$ is negative on it. Because
    $f$ is non-zero and has a fixed sign on any connected component of $S$, we
    conclude that there can only be at most two distinct connected components of
    $S$ having $\bxi$ in their boundary. 
\end{proof}

Let us now introduce some notation. Suppose that we have $\bxi \in V \cap
\realfield^n$ and $\bma \in \ratfield^n$ non-zero. We define, for a new variable
$\varU$, $\bell_{\bxi, \bma} \coloneqq \bxi + \bma\varU \in
\realfield[\varU]^n$. This vector of univariate polynomials defines, as we let
$\varU$ vary in $\realfield$, the set of points lying on the affine line passing
through $\bxi$ in the direction of $\bma$, as illustrated below.

\begin{figure}[H]
\centering
\includegraphics{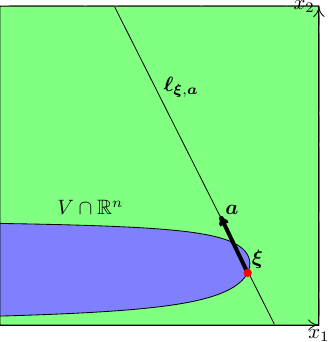}
\caption{Illustration of $\bell_{\bxi, \bma}$ on a two-dimensional example.}
\label{fig:line}
\end{figure}

Let us now define the univariate polynomial $f_{\bxi,\bma}(\varU) \coloneqq
f(\bxi + \bma\varU) \in \realfield[\varU]$. By construction, the roots of
$f_{\bxi,\bma}$ are the relative distances between $\bxi$ and the intersection
points of $V \cap \realfield^n$ and $\bell_{\bxi,\bma}$. In particular, $0$ is a
root of $f_{\bxi,\bma}$, corresponding to $\bxi$ itself. Let us denote by
$\{\zeta_1, \dots, \zeta_r\}$ the non-zero roots of $f_{\bxi,\bma}$, when any
exists. Finally, if $\bxi \in V$, let us define $\T_{\bxi}V$ as the tangent
hyperplane to $V$ at $\bxi$, a well-defined object when $f$ satisfies \assumS.
We use the notation $\left<\bma,\bmb\right>$ to indicate the usual dot product
of two vectors $\bma,\bmb \in \realfield^n$.

\begin{lemma}\label{lem:lineintersectVA}
    Suppose that $f$ satisfies \emph{\assumS}. Let $C$ be a connected component
    of $S$. Suppose that $\bxi \in \overline{C}\setminus C$, and that
    $\bell_{\bxi, \bma}$ does not lie in $\T_{\bxi}V$. Let us define $\delta
    \coloneqq \min\{1,\min_{1\leq i \leq r}|\zeta_i|\}$. Then, for any $0 < z <
    \delta$, either $\bxi + z\bma$ or $\bxi - z\bma$ lies in $C$.
\end{lemma}
\begin{proof}
    Let $C$ and $\bxi$ be as in the statement. Since we suppose that
    $\bell_{\bxi, \bma} \not\subset \T_{\bxi}V$, we have $\left<\Jac_\bxi(f),
    \bma\right> \neq 0$. Suppose without loss of generality that
    $\left<\Jac_\bxi(f), \bma\right> > 0$ (otherwise, consider
    $f_{\bxi,-\bma}$). By Taylor expansion of $f_{\bxi,\bma}$ at $0$, we have
    \[f_{\bxi,\bma}(\varU) = f_{\bxi,\bma}(0) + \varU f_{\bxi,\bma}'(0) +
    o(\varU^2) = 0 + \varU \left<\Jac_\bxi(f), \bma\right> + o(\varU^2) = \varU
    \left<\Jac_\bxi(f), \bma\right> + o(\varU^2).\]
    In particular, there exists a sufficiently small positive $\delta \in
    \realfield$ such that, for any $\varU \in \ ]-\delta, \delta[$,
    $f_{\bxi,\bma}(\varU)$ and $\varU$ have the same sign, since
    $\left<\Jac_\bxi(f), \bma\right>$ is a positive constant.

    Therefore, if we suppose that we have computed such a $\delta \in
    \realfield$, then, for any $0<z<\delta$, $f_{\bxi,\bma}(z) > 0$ and
    $f_{\bxi,\bma}(-z) < 0$, and in particular $\bxi + z\bma$ and $\bxi - z\bma$
    lie in different connected components of $S$ that contain $\bxi$ in their
    boundary. Now, because $f$ satisfies \assumS, we can apply
    \Cref{lem:twoconnectedcomp} to conclude that there are at most two distinct
    connected components of $S$ which have $\bxi$ in their boundary. Since there
    are exactly two such components, and since $C$ is one of them by definition,
    we conclude that either $\bxi + z\bma$ or $\bxi - z\bma$ lies in $C$.
    
    It therefore remains to show that the real number $m \coloneqq
    \min\{1,\min_{1\leq i \leq r}|\zeta_i|\}$ as defined in the statement is a
    suitable choice for $\delta$. By definition, $m > 0$, and $0$ is the only
    root of $f_{\bxi,\bma}$ in the interval $]-m, m[$. In particular, by
    continuity, $f_{\bxi,\bma}$ has constant sign on $]-m,0[$ and on $]0,m[$,
    which means that $m$ is a suitable choice for $\delta$, as required.
\end{proof}

\begin{lemma}\label{lem:subrisolcorrect}
    Suppose that $f$ satisfies \emph{\assumS}. Suppose that $\scrQ = (w, \bmv,
    \nu)$ is a zero-dimensional rational parametrisation of a non-empty subset
    $Q \subset V$. Let $\bma \in \ratfield^n$ such that, for each $\bxi \in Q$,
    $\bell_{\bxi, \bma}$ is not tangent to $V$ at $\bxi$. Let $\GamSLP$ be a
    straight-line program evaluating $f$. Then \emph{\nameref{subr:isol}}
    applied to $\GamSLP$, $\bma$ and $\scrQ$ correctly returns two
    zero-dimensional rational parametrisations $\scrP^-$ and $\scrP^+$ of
    non-empty sets $P^-$ and $P^+$, such that $P^- \cup P^+$ intersects every
    connected component of $S$ whose boundary contains a point of $Q$, and such
    that $f$ does not vanish at any point of $P^- \cup P^+$.
\end{lemma}
\begin{proof}
    Let us denote by $d$ the degree of $f$. Let $t_1 < \dots < t_r$ be the real
    roots of $w(\varT)$. By definition of a rational parametrisation, the set
    \(Q\) is equal to \(\left\{\left(\frac{1}{w'(t_i)}\bmv(t_i)\right)_{1 \leq i
    \leq r}\right\}\).
    Let us define $\bxi_i \coloneqq \frac{1}{w'(t_i)}\bmv(t_i)$. Following our
    previous notation, the set
    \[\left\{\bxi_i + \varU\bma\right\}_{1\leq i\leq r} \subset
    (\ratfield[\varU])^n\]
    is therefore the set of parametrisations of the lines $\bell_{\bxi_i,
    \bma}$, and the set of intersection points of these lines with $V \cap
    \realfield^n$ is given by the real roots of the polynomials
    $f_{\bxi_i,\bma}(\varU) \in \realfield[\varU]$ for $1\leq i\leq r$. Let
    $\zeta_{i,1}, \dots, \zeta_{i,s_i}$ be the non-zero real roots of
    $f_{\bxi_i,\bma}$, and define $\delta_i \coloneqq \min\{1,\min_{1\leq j \leq
    s_i}|\zeta_{i,j}|\}$.
    Let us define
    \[f_{\scrQ, \bma}(\varT, \varU) \coloneqq f\left(\frac{1}{w'(\varT)}
    \bmv(\varT) + \varU\bma\right) \in \ratfield(\varT)[\varU]\]
    so that, for all $1\leq i\leq r$, we have $f_{\scrQ, \bma}(t_i, \varU) =
    f_{\bxi_i,\bma}(\varU)$.
    By definition of a rational parametrisation, $w(\varT)$ is squarefree, and
    hence $w'(t_i) \neq 0$ for all $i$. In particular, the polynomial
    $(w'(\varT))^df_{\scrQ, \bma} \in \ratfield[\varT, \varU]$ has, when $\varT$
    is instantiated to $t_i$, the same roots as $f_{\scrQ, \bma}(t_i, \varU)$.
    We work with $(w'(\varT))^df_{\scrQ, \bma}$ in \nameref{subr:isol} instead
    of $f_{\scrQ,\bma}$, as it is a polynomial.

    Because $w \in \ratfield[\varT]$ and $(w'(\varT))^df_{\scrQ, \bma}(\varT,
    \varU) \in \ratfield[\varT, \varU]$, and because $(w'(t_i))^df_{\scrQ,
    \bma}(t_i, \varU)$ is squarefree for all $i$ by construction (otherwise,
    either $f$ would have a square factor, impossible by \assumS, or $w'$, and
    hence $w$, would have one, impossible by the definition of a
    zero-dimensional rational parametrisation), we can apply
    \ttExtensionRealRoots\ on \Cref{lin:subr2:interVl}, which, given these two
    polynomials, returns a tuple $\scrI$ of $r$ lists of isolation intervals for
    each root of $f_{\scrQ, \bma}(t_i, \varU) = f_{\bxi_i,\bma}(\varU)$, for
    $1\leq i\leq r$. Let $E$ be the set of endpoints of all intervals present in
    $\scrI$, and, as done on \Cref{lin:subr2:lambda}, define $\lambda \coloneqq
    (1/2)\min\{1, \min\{|e| : e \in E \wedge e \neq 0\}\}$. By construction, for
    any $1\leq i\leq r$, we have $0 < \lambda < \delta_i$. Note that $\lambda$
    does not depend on $i$.

    Let $C$ be a connected component of $S$ whose boundary contains an element
    of $Q$, say $\bxi_i$ for some $1\leq i\leq r$. Because $f$ satisfies
    \assumS\ and $\bell_{\bxi, \bma}$ is not tangent to $V$ at $\bxi$ by
    assumption, we can apply \Cref{lem:lineintersectVA} to deduce that either
    $\bxi_i + \lambda\bma$ or $\bxi_i - \lambda\bma$ lies in $C$. Therefore, if
    we define the sets 
    \[P^- \coloneqq \left\{\bxi_i - \lambda\bma\right\}_{1\leq i \leq r}\text{
    and }P^+ \coloneqq \left\{\bxi_i + \lambda\bma\right\}_{1\leq i \leq r},\]
    then $P^- \cup P^+$ is non-empty, contains at least one point of every
    connected component of $S$ whose boundary intersects $Q$, and does not
    contain any root of $f$ by definition of $\lambda$. Moreover, the
    parametrisations $\scrP^-$ and $\scrP^+$ computed on
    \Cref{lin:subr2:P-,lin:subr2:P+} respectively encode $P^-$ and $P^+$. Since
    the output of \nameref{subr:isol} on \Cref{lin:subr2:out} is precisely
    $\scrP^-,\scrP^+$, we conclude that, provided $Q$ is not empty,
    \nameref{subr:isol} correctly returns two zero-dimensional rational
    parametrisations of non-empty sets $P^-$ and $P^+$, such that their union
    intersects every connected component of $S$ whose boundary contains a point
    of $Q$, and such that $f$ does not vanish at any point, as required.
\end{proof}

\subsection{Correctness of \texorpdfstring{\nameref{subr:approx}}{}}
\label{subsec:approx}

In this subsection, we show that, given a polynomial $f$ and a zero-dimensional
rational parametrisation $\scrQ$ of a set $Q$ that does not contain any root of
$f$, \nameref{subr:approx} correctly outputs rational approximations $\bmr$ of
$\bmq \in Q$ that lie in the same connected component of $S$ as $\bmq$.

We begin by introducing some notation. Given a polynomial $f \in \realfield[x_1,
\dots, x_n]$, we define $\norm{f}_\infty$ as the maximum of the absolute values
of its coefficients, and, given $\bmx \in \realfield^n$, we define
$\norm{\bmx}_\infty \coloneqq \max_{1 \leq i \leq n}|x_i|$. Given a
point $\bmx \in \realfield^n$ and $r>0$, we also define 
\[\scrB(\bmx, r) \coloneqq \left\{\bmy \in \realfield^n \ : \
\norm{\bmx -  \bmy}_\infty \leq r\right\},\]
that is, the $n$-dimensional box of side-length $2r$ centred at $\bmx$.

\begin{lemma}\label{lem:boxf}
    Let $f \in \ratfield[x_1, \dots, x_n]$ have positive degree and $Q \subset
    \realfield^n$ be non-empty and finite, and assume that $f$ does not vanish
    at any point of $Q$. Let:
    \begin{itemize}
        \item $d = \deg(f)$,
        \item $\scrc = \norm{f}_\infty$,
        \item $\scrM \geq \max_{\bmq \in Q}\left\{\norm{\bmq}_\infty\right\}+1$,
        \item $0 < \scrm \leq \min_{\bmq \in Q}|f(\bmq)|$,
        \item and
        \[\scrd \coloneqq \min\left\{1,\scrm\left(\binom{n+d}{d}^2\scrM^{d}\scrc 
        \right)^{-1}\right\}.\]
    \end{itemize}
    Then $f$ does not vanish at any point of $\scrB(\bmq,\scrd)$, for any $\bmq
    \in Q$.
\end{lemma}
\begin{proof}
    First, we note that, for $d,\scrc,\scrM$ and $\scrm$ as in the statement,
    $\scrd$ is well-defined, since all aforementioned quantities are positive.
    Suppose that we have shown that, for any $\bmq \in Q$, we have 
    \begin{equation}
    \max_{\bmp \in \scrB(\bmq,\scrd)}\left|f(\bmp)-f(\bmq)\right| < \scrm.
    \label{eq:goal}
    \end{equation}
    Then, by definition of $\scrm$ and since $f(\bmq)\neq 0$ by assumption, for
    any $\bmp \in \scrB(\bmq,\scrd)$, we have $ 0 \leq |f(\bmq)| - \scrm <
    |f(\bmp)| < |f(\bmq)| + \scrm$, and in particular $f$ does not vanish on
    $\scrB(\bmq,\scrd)$. It therefore suffices to prove \Cref{eq:goal}.

    We begin by showing \Cref{eq:goal} for $\bmq = \bmzero \coloneqq
    (0,\dots,0)$, and then reduce the general case to this one. Let $f = \sum_{0
    \leq |\balpha| \leq d} f_\balpha \bmx^\balpha$. In this case, we have
    \begin{align}
        \max_{\bmp \in \scrB(\bmzero,\scrd)}\left|f(\bmp)-f(\bmzero)\right| 
        &\leq \max_{\bmp \in \scrB(\bmzero,\scrd)}\sum_{1 \leq |\balpha| \leq d}
        \left|f_\balpha\right|\left|\bmp^\balpha\right| \notag \\
        &\leq \sum_{1 \leq |\balpha| \leq d} \left|f_\balpha\right| \max_{\bmp 
        \in \scrB(\bmzero,\scrd)}\left|\bmp^\balpha\right|
        = \sum_{1 \leq |\balpha| \leq d} \left|f_\balpha\right| \scrd^d
        = \scrd^d \sum_{1 \leq |\balpha| \leq d} \left|f_\balpha\right|\notag\\
        &< \scrd^d \binom{n+d}{d} \norm{f}_\infty \label{eq:nterms}\\
        &\leq \scrd \binom{n+d}{d} \scrc,\label{eq:d<1} \\
        &\leq \scrm\scrM^{-d}\binom{n+d}{d}^{-1} \leq \scrm \label{eq:defd}
    \end{align}
    where \Cref{eq:nterms} follows from the fact that $f$ has
    $\left(\binom{n+d}{d}-1\right)$ non-constant terms, \Cref{eq:d<1} from the
    definition of $\scrc$ and the fact that $0 < \scrd \leq 1$, and
    \Cref{eq:defd} by construction of $\scrd$ and the fact that $\scrM \geq 1$.
    In particular, we obtain
    \[\max_{\bmp \in \scrB(\bmzero,\scrd)}\left|f(\bmp)-f(\bmzero)\right| 
    < \scrm\]
    as required. Suppose now that $\bmq \in Q$ is arbitrary. Then, since
    translation by $\bmq$ preserves distances, we have
    \begin{align*}
        \max_{\bmp \in \scrB(\bmq,\scrd)}\left|f(\bmp)-f(\bmq)\right| &=
        \max_{\bmp \in \scrB(\bmzero,\scrd)}\left|f(\bmp - \bmq)-f(\bmzero)
        \right| < \scrd \binom{n+d}{d} \norm{f(\bmx - \bmq)}_\infty
    \end{align*}
    by applying the previously obtained result for $\bmzero$ up to
    \Cref{eq:nterms} and the fact that $0 < \scrd \leq 1$. Now, we also have
    \[\norm{f(\bmx - \bmq)}_\infty \leq \scrc \scrM^d \binom{n+d}{d},\]
    since $\scrM > \norm{\bmq}_\infty$ and $f$ has at most $\binom{n+d}{d}$
    terms. Therefore, combining this result with the definition of $\scrd$
    yields
    \[\max_{\bmp \in \scrB(\bmq,\scrd)}\left|f(\bmp)-f(\bmq)\right| < \scrd
    \scrc \scrM^d \binom{n+d}{d}^2 = \scrm\]
    which is precisely \Cref{eq:goal} for an arbitrary $\bmq \in Q$, as 
    required.
\end{proof}

With \Cref{lem:boxf} at hand, the correctness of \nameref{subr:approx} readily
follows:

\begin{lemma}\label{lem:subrapproxcorrect}
    Suppose that $\GamSLP$ is a straight-line program evaluating $f \in
    \ratfield[x_1, \dots, x_n]$ of positive degree, and that $\scrQ = (w, \bmv,
    \nu)$ is a zero-dimensional rational parametrisation which parametrises a
    non-empty set $Q = \{\bmq_1, \dots, \bmq_s\} \subset \realfield^n$. For any
    $1\leq i\leq s$, suppose that $f(\bmq_i) \neq 0$, and let $C_i$ be the
    connected component of $S \coloneqq \left\{\bmx \in \realfield^n \ | \
    f(\bmx) > 0\right\}$ in which $\bmq_i$ lies. Finally, let $\scrd$ be as in
    \emph{\Cref{lem:boxf}}.
    Then, \emph{\nameref{subr:approx}} applied to $\GamSLP$ and $\scrQ$
    correctly returns a finite set of rational approximations $\{\bmr_1, \dots,
    \bmr_s\} \subset \ratfield^n$ of $\{\bmq_1, \dots, \bmq_s\}$ such that
    $\bmr_i \in C_i$ for all $1\leq i\leq s$. Moreover, each obtained rational
    point has coordinates of height at most $\log(\scrd)$.
\end{lemma}
\begin{proof}
    Let $d = \deg(f)$ and $\scrc = \norm{f}_\infty$, as computed in
    \Cref{lin:subr3:degree,lin:subr3:maxcoef}, which are both positive as $f$
    has positive degree. The number $\scrm$ computed in \Cref{lin:subr3:minval}
    satisfies by construction $0 < \scrm \leq \min_{1\leq i\leq s}|f(\bmq_i)|$,
    and $\scrM$ computed in \Cref{lin:subr3:maxval} satisfies $\scrM \geq
    \max_{1\leq i\leq s}\left\{\norm{\bmq_i}_\infty\right\} + 1 > 0$. Since $f$
    has positive degree, since $Q$ is finite and non-empty, and since $f$ does
    not vanish on any point of $Q$, we can apply \Cref{lem:boxf} to deduce $f$
    does not vanish at any point of $\scrB(\bmq_i,\scrd)$ for any $\bmq_i \in
    Q$, where $\scrd$ is the quantity computed in \Cref{lin:subr3:dist}.

    In particular, if $C_i$ is the connected component of $S$ in which $\bmq_i$
    lies, any rational approximation $\bmr_i$ of $\bmq_i$ that lies in
    $\scrB(\bmq_i,\scrd)$ will also lie in $C_i$, since the boundary of $C_i$
    lies in $V \cap \realfield^n$, which does not intersect
    $\scrB(\bmq_i,\scrd)$. It therefore suffices to compute a rational
    approximation $\bmr_i$ of $\bmq_i$ that lies in $\scrB(\bmq_i,\scrd)$ for
    all $\bmq_i \in Q$. By definition, to have $\bmr_i \in \scrB(\bmq_i,\scrd)$,
    it must satisfy \(\norm{\bmr_i - \bmq_i}_\infty \leq \scrd\), and hence have
    all its coordinates accurate to precision at least $2^{\log(\scrd)}$.
    
    The set $\scrA$ computed in \Cref{lin:subr3:approx} is precisely a set of
    rational approximations of points of $Q$ accurate to precision
    $2^{\log(\scrd)}$, and therefore contains rational approximations $\{\bmr_1,
    \dots$, $\bmr_s\}$ of $\{\bmq_1, \dots, \bmq_s\}$ such that $\bmr_i \in
    C_i$. Since the output of \nameref{subr:approx} is precisely $\scrA$ on
    \Cref{lin:subr3:out}, we conclude that \nameref{subr:approx} correctly
    returns a set $\{\bmr_1, \dots, \bmr_s\} \subset \ratfield^n$ such that
    $\bmr_i \in C_i$, as required.
\end{proof}

\subsection{Correctness of \texorpdfstring{\nameref{algo:smooth}}{}}
\label{subsec:correctsmooth}

We now prove the correctness of \nameref{algo:smooth}, using the results
obtained in the previous subsections.

\begin{theorem}\label{thm:algo1correct}
    Let $0 < \epsproba < 1$. Suppose that $f$ satisfies assumption
    \emph{\assumS}, that the matrix $\changeofvarmatrix$ chosen in
    \emph{\Cref{lin:gen:A}} satisfies \emph{\assumHa} and \emph{\assumHb}, that
    the point $\genericcoordsigma$ chosen in \emph{\Cref{lin:gen:sigma}}
    satisfies \emph{\assumB}, and that all calls to \emph{\ttRatParam} are
    successful. Let $\GamSLP$ be a straight-line program evaluating $f$. Then
    \emph{\nameref{algo:smooth}} applied to $\GamSLP$ and $\epsproba$ correctly
    returns a finite subset of $\ratfield^n$ which contains at least one point
    per connected component of $S$.
\end{theorem}
\begin{proof}
    If $S$ is empty, there is nothing to prove, so let us suppose that $S \neq
    \emptyset$. If $V \cap \realfield^n = \emptyset$, then $S$ contains a single
    connected component, equal to $\realfield^n$. Since \nameref{algo:smooth}
    always returns at least one point by construction, this case is covered, and
    we can hence assume without loss of generality that $V \cap \realfield^n
    \neq \emptyset$. Moreover, if $f$ is constant, $S$ is also the whole space
    (since we assume it is not empty), and we hence also assume without loss of
    generality that $f$ has positive degree.

    As $S$ and $V \cap \realfield^n$ are not empty, let $C$ be a connected
    component of $S$ and $B$ be a connected component of its boundary, and
    denote by $C^\changeofvarmatrix$ the corresponding connected component of
    $S^\changeofvarmatrix$, and by $B^\changeofvarmatrix$ the corresponding
    connected component of its boundary. Let us assume for now that
    $\pi_1(C^\changeofvarmatrix) \neq \realfield$, which implies
    $\pi_1(B^\changeofvarmatrix) \neq \realfield$. We now show that, in this
    case, the first iteration of the \texttt{for} loop, for $k=1$, computes a
    finite set $\scrA_1 \subset \ratfield^n$ that contains at least one point of
    $C$.

    Because, by assumption, $f$ satisfies \assumS, $\changeofvarmatrix$
    satisfies \assumHa\ and \assumHb, $\genericcoordsigma$ satisfies \assumB,
    and because we assume that all calls to \ttRatParam\ are successful, we can
    apply \Cref{lem:subrcritcorrect} to conclude that \nameref{subr:critpts}
    applied to $\GamSLP, \changeofvarmatrix,1,
    \genericcoordsigma,\left\lceil\log(3n\epsproba^{-1})\right\rceil$, as in
    \Cref{lin:gen:crit}, correctly returns a zero-dimensional rational
    parametrisation $\scrQ_1 = (w_1, \bmv_1, \nu_1)$ of the finite set
    $\changeofvarmatrix(\crit(\pi_1,V^\changeofvarmatrix))$. As $f$ satisfies
    \assumS, we can apply \Cref{lem:boundary} to deduce that $B$ is a connected
    component of $V \cap \realfield^n$. Because $f$ satisfies \assumS\ and
    $\changeofvarmatrix$ satisfies \assumHa, we can apply \Cref{lem:projclosed}
    to deduce that $\pi_1(B^\changeofvarmatrix)$ is closed for the Euclidean
    topology. Since $\pi_1(B^\changeofvarmatrix) \neq \realfield$, we therefore
    deduce that $B^\changeofvarmatrix \cap \crit(\pi_1, V^\changeofvarmatrix)
    \neq \emptyset$. In particular, $\crit(\pi_1, V^\changeofvarmatrix)$ itself
    is non-empty, and hence $w_1 \not\equiv 1$.

    We are thus in the situation of \Cref{lin:gen:iftrivial} of
    \nameref{algo:smooth}. As it is non-empty, let $\bxi \in
    B^\changeofvarmatrix \cap \crit(\pi_1, V^\changeofvarmatrix)$, so that
    $\changeofvarmatrix\bxi$ is one of the points parametrised by $\scrQ_1$.
    Because $\bxi$ is a critical point of $\pi_1$ restricted to
    $V^\changeofvarmatrix$, the line $\ell_{\bxi, \bme_1}^\changeofvarmatrix$,
    defined by its parametrisation \(\bxi + \bme_1\varU\) for $\bme_1 =
    (1,0,\dots,0)^{\T}$ and a new variable $\varU$, is normal to
    $V^\changeofvarmatrix$ at $\bxi$. Therefore, by applying the change of
    variables $\changeofvarmatrix$ and defining $\bma_1 \coloneqq (a_{1,1},
    \dots, a_{n,1})^{\T}$, we deduce that the line $\ell_{\bxi, \bma_1}$,
    defined by its parametrisation $\changeofvarmatrix(\bxi + \bme_1\varU) =
    \changeofvarmatrix\bxi + \bma_1\varU$, is not in the tangent hyperplane
    $\T_{\bxi}V$. 

    Therefore, $\scrQ_1$ is a zero-dimensional rational parametrisation of the
    non-empty set $\changeofvarmatrix(\crit(\pi_1,V^\changeofvarmatrix))$, on
    which $f$ does not vanish, and the line $\ell_{\bxi, \bma_1}$ is not tangent
    to $V$ at $\bxi$. Because $f$ satisfies \assumS\ and $\bma_1$ does not
    depend on the choice of $\bxi$, we can apply \Cref{lem:subrisolcorrect} to
    conclude that \nameref{subr:isol} applied to $\GamSLP$, $\bma_1$ and
    $\scrQ_1$ as in \Cref{lin:gen:Pparam} correctly returns two zero-dimensional
    rational parametrisations $\scrP_1^-$ and $\scrP_1^+$ of non-empty sets
    $P_1^-$ and $P_1^+$, whose union intersects every connected component having
    a point of $\changeofvarmatrix(\crit(\pi_1,V^\changeofvarmatrix))$ in its
    boundary. In particular, as $\bxi \in B^\changeofvarmatrix \cap \crit(\pi_1,
    V^\changeofvarmatrix)$, we have $\changeofvarmatrix\bxi \in B \cap
    \changeofvarmatrix\crit(\pi_1, V^\changeofvarmatrix)$, and hence $P_1^- \cup
    P_1^+$ contains at least one point of $C$. Moreover,
    \Cref{lem:subrisolcorrect} also implies that $f$ does not vanish at any
    point of $P^- \cup P^+$.

    Since neither $P_1^-$ and $P_1^+$ are empty, and $f$ is non-zero at every
    point of $P_1^- \cup P_1^+$, we can apply \Cref{lem:subrapproxcorrect} to
    conclude that \nameref{subr:approx} applied to $\GamSLP$ and $\scrP^-$ and
    $\scrP^+$, as in \Cref{lin:gen:approx-,lin:gen:approx+}, correctly returns
    sets $\scrA_1^-$ and $\scrA_1^+$ of rational approximations of points of
    $P_1^-$ and $P_1^+$ that lie in the same connected components of $S$ as
    their exact counterparts. In particular, since $P_1^- \cup P_1^+$ contains a
    point of $C$, so does $\scrA_1 \coloneqq \scrA_1^- \cup \scrA_1^+$.
    Therefore, the set $\scrA_1$ computed on \Cref{lin:gen:Ak} of
    \nameref{algo:smooth} is a finite subset of $\ratfield^n$ that contains at
    least one point of $C$, as required.

    Let us now suppose that $\pi_1(C^\changeofvarmatrix) = \realfield$. In this
    case, $C^\changeofvarmatrix$ has non-empty intersection with every fibre of
    $\pi_1$, and, in particular, it therefore suffices to compute a point of
    $C^\changeofvarmatrix \cap \pi_1^{-1}(\sigma_1)$ to obtain a point of
    $C^\changeofvarmatrix$, and hence of $C$. If we suppose now that
    $\pi_2\left(C^\changeofvarmatrix \cap \pi_1^{-1}(\sigma_1)\right) \neq
    \realfield$, then this implies $\pi_2\left(B^\changeofvarmatrix \cap
    \pi_1^{-1}(\sigma_1)\right) \neq \realfield$. Because $f$ satisfies \assumS,
    $\changeofvarmatrix$ satisfies \assumHa\ and \assumHb, $\genericcoordsigma$
    satisfies \assumB, and because we assume that all calls to \ttRatParam\ are
    successful, we can apply the same argument as before to deduce that the
    \texttt{for} loop in the $k=2$ case computes a finite set $\scrA_2 \subset
    \ratfield^n$ that contains at least one point of $C$. Otherwise, we have
    $\pi_2\left(C^\changeofvarmatrix \cap \pi_1^{-1}(\sigma_1)\right) =
    \realfield$, and it hence suffices to consider $C^\changeofvarmatrix \cap
    \pi_2^{-1}(\sigma_1, \sigma_2)$. We repeat this procedure, incrementing $k$
    each time.

    By the end of the $n$-th and last iteration of the \texttt{for} loop, we
    have therefore obtained finite sets $\scrA_1, \dots, \scrA_n \subset
    \realfield^n$ such that, for any connected component $C$ of $S$, either
    $\bigcup_{i=1}^n \scrA_i$ contains a point of $C$, or $C$ satisfies
    $\pi_n\left(C^\changeofvarmatrix \cap \pi_{n-1}^{-1}(\genericcoordsigma)
    \right) = \realfield$. In the latter case, this implies in particular that
    $C^\changeofvarmatrix \cap \pi_n^{-1}(\sigma_1, \dots, \sigma_{n-1},0) =
    C^\changeofvarmatrix \cap (\sigma_1, \dots, \sigma_{n-1},0)$ is non-empty,
    and therefore that $(\sigma_1, \dots, \sigma_{n-1},0) \in
    C^\changeofvarmatrix$ and hence that $\changeofvarmatrix(\sigma_1, \dots,
    \sigma_{n-1},0)^{\T} \in C$. We therefore conclude that, for any connected
    component $C$ of $S$, the finite set $\bigcup_{i=1}^n \scrA_i \cup
    \{\changeofvarmatrix(\sigma_1, \dots, \sigma_{n-1},0)^{\T}\} \subset
    \ratfield^n$ contains at least one point of $C$. As this set is precisely
    the output of \nameref{algo:smooth} on \Cref{lin:gen:output}, we conclude
    that \nameref{algo:smooth} correctly returns a finite subset of
    $\ratfield^n$ that contains at least one point per connected component of
    $S$, as required.
\end{proof}

\section{Probability of success}\label{sec:proba}

In this section, we compute the probability of success of
\nameref{algo:smooth} for a given polynomial $f$ that satisfies \assumS\ and a
given $0<\epsproba<1$. 

Our probabilistic analysis follows from ideas and results by
\cite{EGS23,EGS20,KPS01}, where we aim to find the degree of a polynomial
describing a hypersurface which contains the closed Zariski sets we wish to
avoid. Once this degree obtained, we can compute the probability of a randomly
chosen point to not cancel this polynomial, and hence of not being in the
aforementioned closed Zariski sets, using the so-called
\emph{DeMillo-Lipton-Schwartz-Zippel Lemma} \cite{Schwartz80}, which we recall
below for the reader's convenience.

\noindent DeMillo-Lipton-Schwartz-Zippel Lemma \cite{Schwartz80}: \emph{Let $R$
be an integral domain and $p \in R[x_1, \dots, x_n]$ be a non-zero polynomial.
Let $T \subset R$ be finite such that $|T|>\deg(p)$, and let $(x_1, \dots, x_n)$
be a tuple of values selected uniformly at random from $T$. Then the probability
that $p(x_1, \dots, x_n) = 0$ is bounded above by $\frac{\deg(p)}{|T|}$}.

By \Cref{thm:algo1correct}, as the given $f$ is assumed to satisfy \assumS,
\nameref{algo:smooth} succeeds when we have correctly chosen $\changeofvarmatrix
\in \ratfield^{n\times n}$ satisfying \assumHa\ and \assumHb,
$\genericcoordsigma \in \ratfield^{n-1}$ satisfying \assumB, and when all calls
to \ttRatParam\ are successful. The probabilities of correctly satisfying
\assumHa\ and \assumB\ were already studied in \cite{EGS20}, and we recall its
key results:

\noindent \cite[Theorem 2.1]{EGS20}: \emph{There exists a non-zero polynomial $h
\in \compfield[\frakA]$ of degree at most $5n^3(2d)^{2n}$ such that, if
$h(\changeofvarmatrix) \neq 0$, then $\changeofvarmatrix$ satisfies} \assumHa.

\noindent \cite[Theorem 2.2]{EGS20}: \emph{There exists a non-zero polynomial $g
\in \compfield[\frakS]$ of degree at most $nd^{2n}$ such that, if
$g(\genericcoordsigma) \neq 0$, then $\genericcoordsigma$ satisfies} \assumB.

Combining these results with the DeMillo-Lipton-Schwartz-Zippel Lemma yields the
following propositions. Note that all these proofs follow a similar structure to
\cite[Section 5]{EGS20}.

\begin{proposition}\label{prop:Aproba} 
    Suppose that $0 < \epsproba < 1$ and that the $n \times n$ matrix
    $\changeofvarmatrix$ is randomly drawn with coefficients within the set
    $\setD \coloneqq \left\{1, \dots, \left\lceil 3\epsproba^{-1}
    \left(5n^3(2d)^{2n}+\frac{n^2-n}{2}\right)\right\rceil \right\}$. Then, the
    probability that $\changeofvarmatrix$ satisfies \emph{\assumHa} and
    \emph{\assumHb} is at least $1 - \epsproba/3$.
\end{proposition}
\begin{proof}
    Let $\epsproba$ and $\setD$ be as in the statement. For $\changeofvarmatrix$
    to satisfy \assumHb, the determinant of each bottom right $(n-k) \times
    (n-k)$ submatrix of $\changeofvarmatrix$ must be non-zero, for $1 \leq k
    \leq n$. As each such determinant is a non-zero polynomial of degree at most
    $n-k$ in the variables $\frakA$, we conclude, by taking the product of
    all these polynomials, that there exists a non-zero polynomial $p \in
    \compfield[\frakA]$ of degree at most $\sum_{k=1}^n (n-k) =
    \frac{n^2-n}{2}$ such that $\changeofvarmatrix$ satisfies \assumHb\ whenever
    $p(\changeofvarmatrix) \neq 0$.
    
    On the other hand, by \cite[Theorem 2.1]{EGS20}, there exists a non-zero $h
    \in \compfield[\frakA]$ of degree at most $5n^3(2d)^{2n}$ such that,
    whenever $h(\changeofvarmatrix) \neq 0$, $\changeofvarmatrix$ satisfies
    \assumHa. Therefore, there exists a non-zero polynomial $q \coloneqq ph \in
    \compfield[\frakA]$ of degree at most $5n^3(2d)^{2n} +
    \frac{n^2-n}{2}$ such that, whenever $q(\changeofvarmatrix) \neq 0$,
    $\changeofvarmatrix$ satisfies \assumHa\ and \assumHb.
    
    Thus, if we pick $\changeofvarmatrix$ with arbitrary coefficients
    in $\setD$, by the DeMillo-Lipton-Schwartz-Zippel Lemma, since
    $0<\epsproba<1$, the probability that a randomly chosen $\changeofvarmatrix$
    cancels $q$ is
    \[\probaP[q(\changeofvarmatrix) = 0] \leq \frac{\deg q}{|\setD|} =
    \frac{5n^3(2d)^{2n} + \frac{n^2-n}{2}}{\left\lceil
    3\epsproba^{-1}\left(5n^3(2d)^{2n}+\frac{n^2-n}{2}\right)\right\rceil} \leq
    \frac{\epsproba}{3}.\]
    Hence, the probability that a randomly drawn $\changeofvarmatrix \in
    \setD^{n \times n}$ satisfies \assumHa\ and \assumHb\ is at least $1-
    \epsproba/3$, as required.
\end{proof}

\begin{proposition}\label{prop:Sigmaproba} 
    Suppose that $0 < \epsproba < 1$ and that the point $\genericcoordsigma$ is
    randomly drawn with coefficients within the set $\setE \coloneqq \left\{1,
    \dots, \left\lceil 3\epsproba^{-1}nd^{2n} \right\rceil \right\}$. Then the
    probability that $\genericcoordsigma$ satisfies \emph{\assumB}\ is at least
    $1 - \epsproba/3$.
\end{proposition}
\begin{proof}
    Let $\epsproba$ and $\setE$ be as in the statement. By \cite[Theorem
    2.2]{EGS20}, there exists non-zero polynomial $g \in \compfield[\frakS]$ of
    degree at most $nd^{2n}$ such that, if $g(\genericcoordsigma) \neq 0$, then
    $\genericcoordsigma$ satisfies \assumB. By the
    DeMillo-Lipton-Schwartz-Zippel Lemma, since $0<\epsproba<1$, the probability
    that a randomly chosen $\genericcoordsigma$ cancels $g$ is
    \[\probaP[g(\genericcoordsigma) = 0] \leq \frac{\deg g}{|\setE|} =
    \frac{nd^{2n}}{\left\lceil 3\epsproba^{-1}nd^{2n}\right\rceil} 
    \leq \frac{\epsproba}{3}.\]
    Therefore, the probability that a randomly drawn $\genericcoordsigma \in
    \setE^{n-1}$ satisfies \assumB\ is at least $1 - \epsproba/3$.
\end{proof}

On the other hand, the probability of success of a single call to
\ttRatParam\ was studied in \cite[Theorem 1]{SS18}, and is at least $21/32$.
Combining all these results yields:

\begin{theorem}\label{thm:alg1proba} 
    Suppose that $0<\epsproba<1$, that $f \in \ratfield[x_1, \dots, x_n]$
    satisfies \emph{\assumS}, and that $\GamSLP$ is a straight-line program
    evaluating $f$. Then, the probability of success of
    \emph{\nameref{algo:smooth}} applied to $\GamSLP$ and $\epsproba$ is at
    least $1-\epsproba$.
\end{theorem}
\begin{proof}
    Because $f$ satisfies \assumS, by \Cref{thm:algo1correct},
    \nameref{algo:smooth} is correct when the matrix $\changeofvarmatrix$ chosen
    in \Cref{lin:gen:A} satisfies \assumHa\ and \assumHb, the point
    $\genericcoordsigma$ chosen in \Cref{lin:gen:sigma} satisfies \assumB, and
    when all calls to \ttRatParam\ are correct.

    As, in \Cref{lin:gen:A}, the matrix $\changeofvarmatrix$ is picked with
    randomly drawn coefficients in the set $\setD \coloneqq \left\{1, \dots,
    \left\lceil 3\epsproba^{-1}
    \left(5n^3(2d)^{2n}+\frac{n^2-n}{2}\right)\right\rceil \right\}$, we can
    apply \Cref{prop:Aproba} to deduce that it has a probability of at least $1
    - \epsproba/3$ of satisfying \assumHa\ and \assumHb.
    Similarly, in \Cref{lin:gen:sigma}, the point $\genericcoordsigma$ is picked
    with randomly draw coefficients in the set $\setE \coloneqq \left\{1,
    \dots, \left\lceil 3\epsproba^{-1}nd^{2n} \right\rceil \right\}$, we can
    apply \Cref{prop:Sigmaproba} to deduce that it has a probability of at least
    $1-\epsproba/3$ of satisfying \assumB.
    Finally, the probability of success of a single call to \ttRatParam\ is
    at least $\frac{21}{32} > \frac{1}{2}$, by \cite[Theorem 1]{SS18}, and we
    need to call it on $n$ distinct systems, once per iteration of the
    \texttt{for} loop. Hence, if we call \ttRatParam\ 
    $\left\lceil\log(3n\epsproba^{-1})\right\rceil$ times per iteration, the
    probability that it returns a correct result in every iteration is therefore
    at least
    \[\left(1 - 2^{-\left\lceil\log(3n\epsproba^{-1})\right\rceil}\right)^{n}
    \geq \left(1 - \frac{\epsproba}{3n}\right)^n \geq 1 - \frac{\epsproba}{3}\]
    by Bernoulli's inequality, since $0 < \epsproba < 1$.
    
    Therefore, the probability that $\changeofvarmatrix$ satisfies \assumHa\ and
    \assumHb, $\genericcoordsigma$ satisfies \assumB, and that all calls to
    \ttRatParam\ are correct, since $0<\epsproba<1$, is at least
    \[\left(1-\frac{\epsproba}{3}\right)\left(1-\frac{\epsproba}{3}\right)\left(
    1-\frac{\epsproba}{3}\right) \geq 1 - \epsproba.\]
    Because $f$ satisfies \assumS, by
    applying \Cref{thm:algo1correct}, we therefore conclude that
    \nameref{algo:smooth} correctly returns a finite subset of $\ratfield^n$
    containing at least one point per connected component of $S$ with
    probability at least $1-\epsproba$, as required.
\end{proof}

\section{Bit complexity: proofs of \texorpdfstring{\Cref{thm:mainmulti,thm:mainresult,cor:mainresult}}{}}\label{sec:comp}

We now analyse the bit complexity of \nameref{algo:smooth}. We use $O$ and
$\widetilde{O}$ to indicate the bit cost of an operation with and without
logarithmic factors respectively. For the remainder of this section, we suppose,
to simplify the statements of the results, that the degree of the input
polynomial is at least $2$, since computing points per connected components in
the linear case reduces to elementary computations.

\subsection{Bit Complexity of \texorpdfstring{\nameref{subr:critpts}}{}}
\label{subsec:critptscomp}

We begin with the bit complexity of \nameref{subr:critpts}. As its core is the
subroutine \ttRatParam, we recall for convenience the main result of \cite{SS18}
regarding its bit complexity:

\noindent \cite[Theorem 1]{SS18}: \emph{Let $(f_1, \dots, f_n)$ be a tuple of
polynomials in $\ratfield[\bmx_1, \dots, \bmx_m]$, where $\bmx_i = (x_{i,1},
\dots, x_{i,n_i})$, and $n = n_1 + \dots + n_m$. Suppose that $\V(f_1, \dots,
f_n)$ is zero-dimensional. Suppose that $\deg_{\bmx_j}(f_i) \leq
\widetilde{d}_{i,j}$, that $\htt(f_i) \leq s_i$, and that there exists a
straight-line program $\GamSLP$ of length $L$ evaluating $(f_1, \dots, f_n)$,
using constants of heights at most $b$. Then there exists an algorithm
\emph{\ttRatParam} which produces either:
\begin{itemize}
    \item a zero-dimensional rational parametrisation of the non-singular
    complex solutions to $f_1 = \dots = f_n = 0$,
    \item or a zero-dimensional rational parametrisation of a subset of the
    solutions,
    \item or \emph{\texttt{fail}}.
\end{itemize}
The first outcome occurs with probability at least $21/32$, and in any case, the
algorithm performs}
\[\widetilde{O}\left(Lb + \scrC_{\bmn}(\bmd) \scrH_{\bmn}(\bmeta, \bmd)
\left(L + n\mathfrak{d} + n^2\right)n(\log(s)+n)\right)\]
\emph{bit operations, where
\begin{itemize}
    \item $\mathfrak{d} \coloneqq \max_{1\leq i\leq n}\left(\widetilde{d}_{i,1}
    + \dots + \widetilde{d}_{i,m}\right)$,
    \item $s \coloneqq \max_{1\leq i\leq n}(s_i)$,
    \item $\bmeta \coloneqq \left(s_i + \sum_{j=1}^m
    \log(n_j+1)\widetilde{d}_{i,j} \right)_{1\leq i\leq n}$,
    \item $\scrC_{\bmn}(\bmd)$ is the sum of the coefficients of the polynomial
    \[\prod_{i=1}^n \left(\widetilde{d}_{i,1}\theta_1 + \dots +
    \widetilde{d}_{i,m}\theta_m\right) \mod \left<\theta_1^{n_1+1}, \dots,
    \theta_m^{n_m+1}\right>,\]
    \item and $\scrH_{\bmn}(\bmeta, \bmd)$ is the sum of the  coefficients of
    the polynomial
    \[\prod_{i=1}^n \left(\eta_i\zeta + \widetilde{d}_{i,1}\theta_1 + \dots +
    \widetilde{d}_{i,m} \theta_m\right) \mod \left<\zeta^2, \theta_1^{n_1+1},
    \dots, \theta_m^{n_m+1}\right>.\]
\end{itemize}
The algorithm calls an oracle $\scrO$, taking as input a positive integer $B$
and returning a prime number in $\{B+1, \dots, 2B\}$, with the input
$s\mathfrak{d}^{O(n)}$. Moreover, upon success, any polynomial in the output has
degree at most $\scrC_{\bmn}(\bmd)$ and height in}
$\widetilde{O}\big(\scrH_{\bmn}(\bmeta, \bmd) + n\scrC_{\bmn}(\bmd)\big)$.

This result gives us the bit complexity \ttRatParam\ subroutine, provided that
we know the length of the input straight-line program. This is given to us by
the following result:

\noindent Baur-Strassen Theorem \cite[Theorem 1]{BS83}: \emph{Let $f \in
\ratfield[x_1, \dots, x_n]$, and suppose that there exists a straight-line
program $\GamSLP$ of length $L$ evaluating $f$. Then, there exists a
straight-line program $\LamSLP$ of length at most $3L$ evaluating the
polynomials $\left(f, \frac{\partial f}{\partial x_1}, \dots, \frac{\partial
f}{\partial x_n}\right)$.}

We now have all the tools to express the bit complexity of
\nameref{subr:critpts}. To apply \cite[Theorem 1]{SS18} in its full generality,
we suppose in this section that the input polynomial $f$ has a multi-homogeneous
structure, that is, $f \in \ratfield[\bmx_1, \dots, \bmx_m]$, where $\bmx_i =
\left(x_{1 +\sum_{j=1}^{i-1}n_j}, \dots, x_{n_i + \sum_{j=1}^{i-1}n_j}\right)$,
$n = n_1 + \dots + n_m$, and $\deg_{\bmx_i}(f) \leq d_i$. We moreover assume
that $d_1 \leq \cdots \leq d_m$, since this can always be ensured by simply
re-labelling the groups of variables. 
To simplify notation, for given $\changeofvarmatrix \in \ratfield^{n\times n}$,
$\genericcoordsigma \in \ratfield^{n-1}$, and $f$ as above, we also define:
\begin{itemize}
    \item $d \coloneqq \deg(f)$,
    \item $\delta_{c,\ell,k} \coloneqq \deg_{{\bmx_c}}\Big(
    \frac{\partial f}{\partial x_{\ell}} + \sum_{i=1}^k
    \bigg(\sum_{j=1}^{n-k}a_{i,(j+k)}\lowEnt_{j,\ell-k,k}\bigg) \frac{\partial
    f}{\partial x_i}\Big)$, for $1 \leq c \leq m$, $1 \leq k \leq n$, and $\ell
    > k$,
    \item $\tau \coloneqq \htt(f)$,
    \item $\alpha \coloneqq \max_{1\leq i,j \leq n}\htt(a_{i,j})$,
    \item $\sigma \coloneqq \max_{1\leq i\leq n-1} \htt(\sigma_i)$,
    \item $\approxheight \coloneqq n\alpha +\tau+d\log(n)$,
    \item $\multitotaldegree$ as the sum of the coefficients of the
    polynomial
    \[(d_1\theta_1 + \dots + d_m\theta_m)\prod_{i=2}^{n}
    \left(\delta_{1,i,1}\theta_1 + \dots + \delta_{m,i,1}\theta_m\right) \mod
    \left<\theta_1^{n_1+1}, \dots, \theta_m^{n_m+1}\right>,\]
    \item $\multitotalheight$ as the sum of the coefficients of the polynomial
    \begin{align*}
        &(\approxheight\zeta + d_1\theta_1 + \dots + d_m\theta_m)
        \prod_{i=2}^n \left(\approxheight\zeta + \delta_{1,i,1}\theta_1 + 
        \dots + \delta_{m,i,1}\theta_m\right) \\
        &\mod \left<\zeta^2, \theta_1^{n_1+1}, \dots, \theta_m^{n_m+1}\right>.
    \end{align*}
\end{itemize}

\begin{lemma}\label{lem:subrcritcomp}
    Suppose that $f \in \ratfield[\bmx_1, \dots, \bmx_m]$ satisfies
    \emph{\assumS}, $\changeofvarmatrix \in \ratfield^{n\times n}$ satisfies
    \emph{\assumHa} and \emph{\assumHb}, and that $\genericcoordsigma \in
    \ratfield^{n-1}$ satisfies \emph{\assumB}. Let $1\leq k\leq n$ and $r \in
    \Nintegers \setminus \{0\}$, and suppose that $\sigma < \alpha$. Let
    $\GamSLP$ be a straight-line program of length $L$ evaluating $f$.
    Then \emph{\nameref{subr:critpts}} applied to $\GamSLP$,
    $\changeofvarmatrix$, $k$, $\genericcoordsigma$ and $r$ performs at most
    \[\widetilde{O}\left(rn^2\multitotaldegree\multitotalheight
    (L+nd+n^3)\right)\]
    bit operations. Moreover, upon success, any polynomial in the output has
    degree at most $\multitotaldegree$ and height in
    $\widetilde{O}\left(\multitotalheight + n\multitotaldegree\right)$.
\end{lemma}
\begin{proof}
    Since $\changeofvarmatrix$ satisfies \assumHb, each $\lowEnt_{i,j,k}$ is
    well-defined. Moreover, by Cramer's rule and Hadamard's inequality, we have 
    \[\htt(\lowEnt_{i,j,k}) \leq \log\left(n^{n/2}\left(2^{\alpha}\right)^n
    \right) \in \widetilde{O}\left(n\log(n) + n\alpha\right) = \widetilde{O}
    \left(n\alpha\right).\]
    We begin by computing the length of the straight-line program $\LamSLP$
    computed in \Cref{lin:subr1:SLP}. Recall that it evaluates the polynomials
    $\calP_k$, which consist of:
    \begin{itemize}
        \item $k-1$ polynomials, $\sum_{i=1}^n \lowEnt_{1,i,0}x_i - \sigma_1,
        \dots, \sum_{i=1}^n \lowEnt_{k-1,i,0}x_i-\sigma_{k-1}$,
        \item $1$ polynomial, $f$,
        \item and $n-k$ polynomials, 
        \[\frac{\partial f}{\partial x_{k+1}} + \sum_{i=1}^k
        \bigg(\sum_{j=1}^{n-k}a_{i,(j+k)}\lowEnt_{j,1,k}\bigg) \frac{\partial
        f}{\partial x_i}, \ \dots, \ \frac{\partial f}{\partial x_n} +
        \sum_{i=1}^k \bigg(\sum_{j=1}^{n-k}a_{i,(j+k)}\lowEnt_{j,n-k,k}\bigg)
        \frac{\partial f} {\partial x_i}.\]
    \end{itemize}
    By Baur-Strassen \cite[Theorem 1]{BS83}, there exists a straight-line
    program of length at most $3L$ evaluating $f$ and its partial derivatives.
    We include in this program the $n^2$ entries of $\changeofvarmatrix$, the
    $n^2$ entries of $\changeofvarmatrix^{-1}$, the $(n-k)^2$ entries of
    $\lowMat_k^{-1}$, and the $n-1$ entries of $\genericcoordsigma$. Finally, to
    obtain the full system of \ttSLPPolar, we perform $2n(k-1)$ operations with
    those entries to obtain the first $k-1$ polynomials, leave $f$ unchanged,
    and perform $(n-k)(k + 2k(n-k))$ operations to obtain the remaining $n-k$
    polynomials. Therefore, the length of $\LamSLP$ is bounded by 
    \[3L + n^2 + n^2 + (n-k)^2 + n-1 + 2n(k-1) + k(n-k)(1+2(n-k)) \in O\big(
    L + n^3\big).\]
    
    We now compute the quantities appearing in the statement of \cite[Theorem
    1]{SS18} for this system:
    \begin{itemize}
        \item len$(\LamSLP) \in \widetilde{O}(L+n^3)$,
        \item $\widetilde{d} = (\widetilde{d}_1, \dots, \widetilde{d}_n)$,
        where:
        \begin{itemize}
            \item $\widetilde{d}_1 = \cdots = \widetilde{d}_{k-1} = (1, \dots,
            1)$,
            \item $\widetilde{d}_k = (d_1, \dots, d_m)$,
            \item $\widetilde{d}_{k+1} = (\delta_{1,k+1,k}, \dots,
            \delta_{m,k+1,k}), \dots, \widetilde{d}_n = (\delta_{1,n,k}, \cdots,
            \delta_{m,n,k})$,
        \end{itemize}
        \item $\bms \in \widetilde{O}\big(\left(n\alpha, \dots, n\alpha, \tau,
        n\alpha + \tau+\log(d), \dots, n\alpha + \tau+\log(d)\right)\big)$,
        since $\sigma < \alpha$,
        \item $\mathfrak{d} = d$,
        \item $b \in \widetilde{O}(n\alpha + \tau+\log(d)) \in
        \widetilde{O}(\approxheight)$,
        \item $\bmeta \in \widetilde{O}\big((n\alpha, \dots, n\alpha, \tau +
        d\log(n), \approxheight,\dots,\approxheight)\big)$.
    \end{itemize}
    By \cite[Theorem 1]{SS18}, the complexity of \ttRatParam\ is increasing in
    the entries of $\widetilde{\bmd}$, $\bms$ and $\bmeta$, and in this
    particular case, all of their entries are the largest when $k=1$, since we
    assume that $d_1 \leq \dots \leq d_m$. It therefore suffices to compute the
    bit complexity of \nameref{subr:critpts} when $k=1$ to obtain a bound on its
    complexity for all $k$. If $k=1$, the above quantities yield
    \begin{itemize}
        \item $\scrC_{\bmn}(\bmd) = \multitotaldegree$,
        \item $\scrH_{\bmn}(\bmeta, \bmd) \in \widetilde{O}(\multitotalheight)$.
    \end{itemize}
    Since $f$ satisfies \assumS, $\changeofvarmatrix$ satisfies \assumHa\ and
    \assumHb, and $\genericcoordsigma$ satisfies \assumB, we can apply
    \cite[Theorem 2.2]{EGS20} to conclude that the set of complex solutions to
    the system parametrised by $\LamSLP$ is zero-dimensional and regular.
    Therefore, we can apply \cite[Theorem 1]{SS18} to $\LamSLP$ to conclude that
    \ttRatParam\ applied to $\LamSLP$, $\changeofvarmatrix$, $1$,
    $\genericcoordsigma$ and $r$, as done in \Cref{lin:subr1:param}, performs at
    most
    \begin{align*}
        &\phantom{{}={}}\widetilde{O}\left(r\left((L+n^3)\approxheight + 
        \multitotaldegree\multitotalheight(L + n^3+ nd + n^2)n
        (\log(\approxheight)+n)\right)\right) \\
        &\subseteq \widetilde{O}\left(rn^2\multitotaldegree
        \multitotalheight(L+nd+n^3)\right)
    \end{align*}
    bit operations, since we repeat the computation $r$ times and since
    $\multitotalheight \geq \approxheight$ by construction. The statement on
    the degree and height of the output polynomials follows immediately from the
    last part of \cite[Theorem 1]{SS18}.
\end{proof}

\subsection{Bit Complexity of \texorpdfstring{\nameref{subr:isol}}{}}
\label{subsec:isolcomp}

We now analyse the bit complexity of \nameref{subr:isol}. The main tool we use
for this analysis is the following result by \cite{STRZEBONSKI19}:

\noindent \cite[Theorem 16]{STRZEBONSKI19}: \emph{Suppose that $p \in
\ratfield[\varT]$ and $q \in \ratfield[\varT, \varU]$. Let $a = \deg(p)$, $b =
\deg_{\varU}(q)$, $\tau = \htt(p)$ and $\sigma = \htt(q)$. Then, there exists an
algorithm \emph{\ttExtensionRealRoots} which takes $p$ and $q$ as inputs, and
returns isolation intervals for all real roots of $q(t_i,\varU)$, if squarefree,
for all real roots $t_i$ of $p$, which performs at most $\widetilde{O}\left(a
b^3 + a^2b^2(\tau + \sigma)\right)$ bit operations. Moreover, the height of any
computed endpoint is in $\widetilde{O}\left(a b^2(\tau + \sigma)\right)$.}

With this theorem at hand, the complexity of \nameref{subr:isol} readily
follows.

\begin{lemma}\label{lem:subrisolcomp}
    Suppose that $f \in \ratfield[x_1, \dots, x_n]$ satisfies \emph{\assumS},
    has degree $d$ and height $\tau$, and is evaluated by a straight-line
    program $\GamSLP$. Suppose that $\scrQ = (w,\bmv,\nu)$ is a zero-dimensional
    rational parametrisation with polynomials of degree $\delta$ and height
    $\eta$, such that $w \not\equiv 1$, and $f$ vanishes on every point
    parametrised by k. Let $\bma \in \ratfield^n$ satisfy the requirements of
    \emph{\nameref{subr:isol}}, and suppose that $\htt(\bma) \leq \eta$. 
    Then \emph{\nameref{subr:isol}} applied to $\GamSLP$, $\bma$ and $\scrQ$
    performs at most
    \[\widetilde{O}\left(n\delta^2 d^2\tau+n\delta^2
    d^3\eta+n\delta^3 d^3\right)\]
    bit operations. Moreover, the polynomials in the output parametrisations
    $\scrP^-$ and $\scrP^+$ have degree at most $\delta$ and height in
    $\widetilde{O}\big(\delta d^2\tau + \delta d^3 \eta + \delta^2 d^3\big)$.
\end{lemma}
\begin{proof}
    The polynomial $w(\varT)$ has by assumption degree $\delta$ and height
    $\eta$. By Mignotte's bound \cite[Corollary 6.33]{MCA}, the polynomial
    $\left(w'(\varT)\right)^d$ has height
    \[d(\log(\delta - 1) + \eta) + \log\left(\left(\sqrt{d\delta+1}\right)
    2^{d(\delta-1)} \binom{d}{\lfloor d/2 \rfloor}\right) \in
    \widetilde{O}\left(d(\eta + \delta)\right).\] Therefore, the polynomial
    $\left(w'(\varT)\right)^df\left(\frac{1} {w'(\varT)} \bmv(\varT) +
    \varU\bma\right)$ has degree $d$ in $\varU$, and has height in
    \[\widetilde{O}\left(\log(d(\delta-1)+1) + d(\eta + \delta) + 
    (\tau + d(\eta + \delta))\right) = \widetilde{O}\left(\tau + d(\eta +
    \delta)\right).\]
    Therefore, by \cite[Theorem 16]{STRZEBONSKI19},
    \ttExtensionRealRoots~applied to the polynomials $w(\varT) \in
    \ratfield[\varT]$ and $\left(w'(\varT)\right)^df\left(\frac{1} {w'(\varT)}
    \bmv(\varT) + \varU\bma\right) \in \ratfield[\varT, \varU]$, as done in
    \Cref{lin:subr2:interVl}, performs at most
    \[\widetilde{O}\left(\delta d^3 + \delta^2 d^2(\eta + \tau + d(\eta +
    \delta))\right) = \widetilde{O}\left(\delta^2 d^2 \tau + \delta^2
    d^3\eta + \delta^3 d^3 \right)\]
    bit operations. It returns at most $2d\delta$ endpoints, of height
    \[\widetilde{O}\left(\delta d^2 (\eta + \tau + d(\eta + \delta))\right) =
    \widetilde{O}\left(\delta d^2\tau + \delta d^3 \eta + \delta^2 d^3\right).\]
    Therefore, the rational $\lambda$ of \Cref{lin:subr2:lambda} is computed in
    at most $\widetilde{O}(d\delta)$ bit operations, and also has height in
    $\widetilde{O}\left(\delta d^2\tau + \delta d^3 \eta + \delta^2d^3\right)$.
    Hence, the parametrisations $\scrP^-$ and $\scrP^+$ computed in
    \Cref{lin:subr2:P-,lin:subr2:P+} have polynomials of degree at most $\delta$
    and height in 
    \[\widetilde{O}\left(\delta d^2\tau + \delta d^3 \eta +
    \delta^2 d^3+ \eta + \log(\delta-1)+ \eta\right) =
    \widetilde{O}\left(\delta d^2\tau + \delta d^3 \eta + \delta^2 d^3\right).\]
    Computing them requires $\widetilde{O}\left(n\delta\right)$ operations
    between rationals having at most this height, and thus costs
    \[\widetilde{O}\left(n\delta^2 d^2\tau+n\delta^2
    d^3\eta+n\delta^3 d^3\right)\] bit operations.

    Therefore, the most expensive steps of \nameref{subr:isol} are
    \Cref{lin:subr2:P-,lin:subr2:P+}, and hence \nameref{subr:isol} performs at
    most $\widetilde{O}\left(n\delta^2 d^2\tau+n\delta^2
    d^3\eta+n\delta^3 d^3\right)$ bit operations, and outputs two
    rational parametrisations with polynomials of degree $\delta$ and height
    $\widetilde{O}\big(\delta d^2\tau + \delta d^3 \eta + \delta^2 d^3\big)$, 
    as required.
\end{proof}

\subsection{Bit Complexity of \texorpdfstring{\nameref{subr:approx}}{}}
\label{subsec:approxcomp}

We now analyse the bit complexity of \nameref{subr:approx}. The main result used
for this section is the following:

\noindent \cite[Theorem 47]{MS21}: \emph{Let $\scrQ$ be a zero-dimensional
rational parametrisation of a set $Q$ with polynomials of degree at most
$\delta$ and height $\eta$, and let $\kappa$ be a positive integer. Then, there
exists an algorithm} \ttApproximationList\ \emph{which, given $\scrQ$ and
$\kappa$, computes approximations of the real points of $Q$ accurate to
precision $2^{-\kappa}$ in $\widetilde{O}\left(\delta^3 + n\delta^2 \eta +
n\delta\kappa\right)$ bit operations.}

It therefore remains to determine a suitable value of $\kappa$ to perform the
approximation, which is the purpose of the following results:

\begin{lemma}\label{lem:coordheight}
    Let $W \subset \realfield^n$ be a zero-dimensional variety defined over
    $\ratfield$, of degree $\delta$ and height $\eta$. Then any non-zero
    coordinate $b$ of any point of $W$ satisfies
    \[\frac{1}{n^{(n/2 + 1)}\delta^n\left(1 + 2^{2\eta +10\log(n+3)\delta}
    \right)} \leq |b| 
    \leq n^{(n/2 + 1)}\delta^n\left(1 + 2^{2\eta +10\log(n+3)\delta}\right).\]
\end{lemma}
\begin{proof}
    As $W$ contains at most $\delta$ points by definition, there exists an
    invertible change of variables matrix $\bm{M} \in \ratfield^{n\times n}$
    with entries of height at most $\log(\delta)$ such that each point of
    $W^{\bm{M}^{-1}} \coloneqq \{\bm{M}w : w \in W\}$ has distinct $i$-th
    coordinates, for any $1\leq i \leq n$. When this is the case, the ideal
    $\I\big(W^{\bm{M}^{-1}}\big)$ satisfies \cite[Assumption 1]{DS04}, and we
    can hence apply \cite[Theorem 1]{DS04} to conclude that, for any $1\leq i
    \leq n$, there exists a polynomial $N_i \in \ratfield[x_i]$ such that
    \begin{align*}
        &\deg(N_i) \leq \delta, \qquad \htt(N_i) \leq \eta + 5\log(n+3)\delta,\\
        &N_i(c) = 0 \iff \exists \bxi \in W^{\bm{M}^{-1}} : \bxi_i = c.
    \end{align*}
    Moreover, by the Cauchy bounds \cite[Lemmata 10.2-10.3]{BPR}, if $N =
    \sum_{j=0}^\delta n_j x_i^j$ with $n_\delta \neq 0$, any non-zero root $c$
    of $N_i$ satisfies
    \[\frac{1}{1+\max_{0 \leq j < \delta}\left|\frac{n_j}{n_\delta}
    \right|} \leq |c| \leq 1+\max_{0 \leq j < \delta}\left|
    \frac{n_j}{n_\delta}\right|.\] 
    Combining these results for all $1\leq i \leq n$ implies that any non-zero
    coordinate $c$ of any point in $W^{\bm{M}^{-1}}$ satisfies
    \[\frac{1}{1+2^{2\eta + 10\log(n+3)\delta}} \leq |c| \leq 1+2^{2\eta +
    10\log(n+3)\delta}.\]
    Now, by Cramer's rule and Hadamard's inequality, since the entries of
    $\bm{M}$ have height at most $\log(\delta)$, the entries of $\bm{M}^{-1}$
    have height at most $\frac{1}{2}n\log(n) + n\log(\delta)$. We hence conclude
    that any non-zero coordinate $b$ of any point in $W$ satisfies
    \[\frac{1}{n^{(n/2 + 1)}\delta^n\left(1 + 2^{2\eta +10\log(n+3)\delta}
    \right)} \leq |b| 
    \leq n^{(n/2 + 1)}\delta^n\left(1 + 2^{2\eta +10\log(n+3)\delta}\right)\]
    as required.
\end{proof}

\begin{corollary}\label{cor:minmaxheightv2}
    Let $W \subset \realfield^n$ be a zero-dimensional variety defined over
    $\ratfield$, of degree $\delta$ and height $\eta$. Let $f \in \ratfield[x_1,
    \dots, x_n]$ be a polynomial of degree $d$ and height $\tau$. Then any
    non-zero real value $b$ that $f$ takes on $W$ satisfies
    \[\frac{1}{n^{(n+3)/2}\delta^{n+1}\left(1 + 4^{d\eta + \tau\delta + 
    6\log(n+4)\delta}\right)} \leq |b| \leq 
    (n+1)^{(n+3)/2}\delta^{n+1}\left(1 + 4^{d\eta + \tau\delta + 6\log(n+4)
    \delta}\right).\]
    In particular, if $\scrQ$ is a zero-dimensional rational parametrisation of
    $W$, \emph{\ttParamMinValue} and \emph{\ttParamMaxValue}, when applied to
    $\scrQ$ and $f$, both return rational numbers having height in
    $\widetilde{O}\big(n + d\eta + \tau\delta\big)$.
\end{corollary}
\begin{proof}
    Let $\varT$ be a new variable. To obtain bounds on the non-zero real values
    that $f$ takes on $W$, it suffices to bound the non-zero real values of the
    variety $W \cap \V(f-\varT)$. Since each point of $W$ corresponds to a
    single value of $\varT$ by construction, this variety is also
    zero-dimensional, and has the same degree $\delta$ as $W$. Moreover, by
    \cite[Corollary 2.10]{KPS01}, we have
    \[\htt(W \cap \V(f-\varT)) \leq d\eta + (\tau + \log(n+1))\delta.\]
    Therefore, by applying \Cref{lem:coordheight} to $W \cap \V(f-\varT)$, we
    conclude that any non-zero real value $b$ that $f$ takes on $W$ satisfies
    \[\frac{1}{n^{(n+3)/2}\delta^{n+1}\left(1 + 4^{d\eta + \tau\delta + 
    6\log(n+4)\delta}\right)} \leq |b| \leq 
    (n+1)^{(n+3)/2}\delta^{n+1}\left(1 + 4^{d\eta + \tau\delta + 6\log(n+4)
    \delta}\right).\]
    By taking logarithms, we conclude that these bounds both have heights in
    \begin{align*}
        &\phantom{{}={}}\widetilde{O}\big((n+3)\log(n+1) + (n+1)\log(\delta) + 
        d\eta + \tau\delta + \log(n+4)\delta\big) \\
        &=\widetilde{O}\big(n + d\eta + \tau\delta\big)
    \end{align*}
    as required.
\end{proof}

From this corollary, a bound on the bit cost of \nameref{subr:approx} readily
follows:

\begin{lemma}\label{lem:subrapproxcomp}
    Let $f \in \ratfield[x_1, \dots, x_n]$ be of degree $d$ and height $\tau$,
    and let $\GamSLP$ be a straight-line program evaluating $f$. Let $\scrQ =
    (w, \bmv, \nu)$ be a zero-dimensional rational parametrisation of a set $Q$,
    with polynomials of degree at most $\delta$ and height $\eta$. Suppose that
    $f$ does not vanish on any point parametrised by $\scrQ$, and let $\gamma
    \coloneqq \htt(Q)$. Then, \emph{\nameref{subr:approx}} applied to $\GamSLP$
    and $\scrQ$ performs at most
    \[\widetilde{O}\left(\delta^3 + n\delta^2 \eta + n^2d\delta + 
    nd^2\delta\gamma + \tau nd\delta^2\right)\]
    bit operations, and returns a finite set of $\delta$ rationals having height 
    in
    $\widetilde{O}\big(nd + d^2\gamma + \tau d\delta\big)$. 
\end{lemma}
\begin{proof}
    By definition, the rational $\scrc$ of \Cref{lin:subr3:maxcoef} has height
    $\tau$. By \Cref{cor:minmaxheightv2}, the rationals $\scrm$ and $\scrM$ of
    \Cref{lin:subr3:minval,lin:subr3:maxval} have heights in $\widetilde{O}\big(
    n\log(\delta) + d\gamma + \tau\delta + \delta\log(n)\big)$. 
    Therefore, $\scrd$ as computed in \Cref{lin:subr3:dist} has height in
    \begin{align*}
        &\phantom{{}={}}\widetilde{O}\left(\tau + (d+1)\big(n + 
        d\gamma + \tau\delta\big) + 2(n+d)\log(n+d)\right) \\
        &=\widetilde{O}\big(nd + d^2\gamma + \tau d\delta\big)
    \end{align*}
    and can be computed in quasi-linear time with respect to its height.
    Therefore, by \cite[Theorem 47]{MS21}, the set $\scrA$ of
    \Cref{lin:subr3:approx} is computed in at most
    \begin{align*}
    &\phantom{{}={}}\widetilde{O}\left(\delta^3 + n\delta^2 \eta + n\delta
    \big(nd + d^2\gamma + \tau d\delta\big)\right) \\
    &= \widetilde{O}\left(\delta^3 + n\delta^2 \eta + n^2d\delta + 
    nd^2\delta\gamma + \tau nd\delta^2\right)
    \end{align*}
    bit operations, and contains rationals having height in
    $\widetilde{O}\big(nd + d^2\gamma + \tau d\delta\big)$, as required.
\end{proof}

\subsection{Bit Complexity of \texorpdfstring{\nameref{algo:smooth}}{}}
\label{subsec:alg1comp}

We now finally analyse the bit complexity of \nameref{algo:smooth}. Recall that
we assume that $d > 1$.

\begin{customthm}{1}\label{thm:alg1comp}
    Suppose that $f \in \ratfield[x_1, \dots, x_n]$ satisfies \emph{\assumS}, is
    of degree $d$ and height $\tau$, and that $\GamSLP$ is a straight-line
    program of length $L$ evaluating $f$. Let $0 < \epsproba < 1$. Then
    \emph{\nameref{algo:smooth}} performs at most
    \[\widetilde{O}\big(\log(1/\epsproba)n^3\multitotaldegree
    \multitotalheight(L+nd+n^3) + nd^5\multitotaldegree^3(\multitotalheight 
    + n\multitotaldegree)\big)\]
    bit operations, with probability of success at least $1-\epsproba$. Upon
    success, it returns a finite subset of $\ratfield^n$ containing at most
    $2n\multitotaldegree + 1$ elements, each of height in
    \[\widetilde{O}\big(d^5\multitotaldegree^2(\multitotalheight +
    n\multitotaldegree)\big).\] 
\end{customthm}
\begin{proof}
    As the complexity of \nameref{subr:critpts} computed in
    \Cref{lem:subrcritcomp}, as well as the degree and height of polynomials in
    the output, are independent of $k$, it suffices to compute the bit
    complexity of a single iteration of the \texttt{for} loop, and multiply the
    result by $n$ to obtain the complexity of \nameref{algo:smooth}. 

    By construction, each entry of the matrix $\changeofvarmatrix$ chosen in 
    \Cref{lin:gen:A} has height 
    \[\alpha \coloneqq \log\left(3\epsproba^{-1}\left(5n^3(2d)^{2n}+
    \frac{n^2-n}{2}\right)\right) \in \widetilde{O}\left(\log(1/\epsproba) + 
    n\log(d)\right)\]
    and each entry of the vector $\genericcoordsigma$ chosen in
    \Cref{lin:gen:sigma} has height no higher than $\alpha$. We can therefore
    apply \Cref{lem:subrcritcomp} to conclude that the bit cost of
    \Cref{lin:gen:crit} is at most
    \begin{align*}
        &\phantom{{}={}}\widetilde{O}\left(\log(1/\epsproba)n^2
        \multitotaldegree\multitotalheight(L+nd+n^3)\right).
    \end{align*}
    Upon success, it returns a zero-dimensional rational parametrisation
    $\scrQ_k = (w_k, \bmv_k, \nu_k)$ of degree at most $\multitotaldegree$ and
    height in \(\widetilde{O}(\multitotalheight + n\multitotaldegree)\). If $w_k
    \equiv 1$, the loop terminates here, so suppose that $w_k \not\equiv 1$. As
    $\alpha \leq \approxheight \leq \multitotalheight \leq \multitotalheight +
    n\multitotaldegree$, we can apply \Cref{lem:subrisolcomp} to conclude that
    the bit cost of \Cref{lin:gen:Pparam} is at most
    \begin{align*}
        &\phantom{{}={}}\widetilde{O}\big(nd\multitotaldegree^2 d^2\tau + 
        n\multitotaldegree^2 d^3(\multitotalheight + n\multitotaldegree) 
        +n\multitotaldegree^3d^3\big) \\
        &= \widetilde{O}\big(nd^3\multitotaldegree^2 (\multitotalheight + 
        n\multitotaldegree) \big),
    \end{align*}
    since $\tau \leq \approxheight \leq \multitotalheight$, and that the output
    parametrisations each have degree bounded by $\multitotaldegree$ and height
    in $\widetilde{O}\big(d^3\multitotaldegree (\multitotalheight +
    n\multitotaldegree)\big)$. By using the arithmetic B\'ezout theorem
    \cite{KPS01}, along with \cite[Lemma 2.6]{KPS01}, we know that the heights
    of the algebraic sets $P^-$ and $P^+$ parametrised by $\scrP^-$ and
    $\scrP^+$ satisfy
    \begin{align*}
        \htt(P^\pm) &\leq n\approxheight\multitotaldegree 
        + (2n+3)\multitotaldegree\log(n+1)
        + \big(d^3 \multitotaldegree (\multitotalheight +
        n\multitotaldegree)\big)\multitotaldegree + 
        5\multitotaldegree\log(n+1) \\
        &\in \widetilde{O}\big(d^3\multitotaldegree^2 (\multitotalheight +
        n\multitotaldegree)\big).
    \end{align*}
    Therefore, by applying \Cref{lem:subrapproxcomp}, we
    conclude that the bit cost of \Cref{lin:gen:approx-,lin:gen:approx+} is at
    most
    \begin{align*}
        &\phantom{{}={}}\widetilde{O}\Big(\multitotaldegree^3 + 
        n\multitotaldegree^2\big(d^3 \multitotaldegree (\multitotalheight +
        n\multitotaldegree)\big) + n^2d\multitotaldegree + 
        nd^2\multitotaldegree \big(d^3 \multitotaldegree^2
        (\multitotalheight + n\multitotaldegree)\big) + 
        \tau nd\multitotaldegree^2\Big) \\
        &=\widetilde{O}\big(nd^5\multitotaldegree^3(\multitotalheight +
        n\multitotaldegree)\big),
    \end{align*}
    and that the output rationals each have height in
    \begin{align*}
        &\phantom{{}={}}\widetilde{O}\left(nd + d^2\big(d^3 \multitotaldegree^2
        (\multitotalheight + n\multitotaldegree)\big) + \tau d\multitotaldegree
        \right)\\
        &= \widetilde{O}\big(d^5\multitotaldegree^2(\multitotalheight +
        n\multitotaldegree)\big)
    \end{align*}
    as required. Therefore, the total bit complexity of \nameref{algo:smooth} is
    at most $n$ times the total complexity of an iteration of the \texttt{for}
    loop, that is, at most
    \begin{align*}
        &\phantom{{}={}}\widetilde{O}\Big(n\big(\log(1/\epsproba)n^2
        \multitotaldegree\multitotalheight(L+nd+n^3) \\
        &\indent \indent + nd^3\multitotaldegree^2 (\multitotalheight + 
        n\multitotaldegree) \\
        &\indent \indent + nd^5\multitotaldegree^3(\multitotalheight +
        n\multitotaldegree)\big)\Big) \\
        &= \widetilde{O}\big(\log(1/\epsproba)n^3\multitotaldegree
        \multitotalheight(L+nd+n^3) + nd^5\multitotaldegree^3(\multitotalheight 
        + n\multitotaldegree)\big)
    \end{align*}
    as required.
\end{proof}

Using \Cref{thm:alg1comp}, we obtain the following result, concerning the case
where $f$ does not have any particular multi-homogeneous structure, that is, $f
\in \ratfield[\bmx_1]$. In this case, to simplify notation, we define
\begin{itemize}
    \item $\delta_i \coloneqq \deg\big(\frac{\partial f}{\partial x_i}\big)$ 
    for $1\leq i \leq n$,
    \item $\totaldegree \coloneqq d\delta_2\cdots\delta_n$,
\end{itemize}
and further assume that $\delta_1 \leq \cdots \leq \delta_n$, since
this can always be ensured by re-labelling the variables.

\begin{customthm}{2}\label{thm:alg1compsimp}
    In the setting of \emph{\Cref{thm:alg1comp}}, when $f \in
    \ratfield[\bmx_1]$, \emph{\nameref{algo:smooth}} performs at most
    \[\widetilde{O}\big(n^2\totaldegree^2(n\log(1/\epsproba) + n^2 + 
    \tau + d)(\log(1/\epsproba)n^2(L+nd+n^3) + 
    d^5\totaldegree^2)\big)\]
    bit operations, with probability of success at least $1-\epsproba$. Upon
    success, it returns a finite subset of $\ratfield^n$ containing at most
    $2n\totaldegree + 1$ elements, each of height in
    \[\widetilde{O}\big(nd^5\totaldegree^3(n\log(1/\epsproba) + n^2 + 
    \tau + d)\big).\]
\end{customthm}
\begin{proof}
    In the particular case where $\bmx_1$ is the only group of variables, we
    have:
    \begin{itemize}
        \item $\alpha \in \widetilde{O}\left(\log(1/\epsproba) + 
        n\log(d)\right)$,
        \item $\approxheight \in \widetilde{O}\left(n\log(1/\epsproba) + n^2 + 
        \tau + d\right)$,
        \item $\deg\left(\frac{\partial f}{\partial x_{\ell}} +
        \sum_{i=1}^k \bigg(\sum_{j=1}^{n-k}a_{i,(j+k)}\lowEnt_{j,\ell-k,k}\bigg)
        \frac{\partial f}{\partial x_i}\right) \leq \delta_{\ell}$ for $\ell >
        k$, since we assume that $\delta_1 \leq \cdots \leq \delta_{\ell}$.
    \end{itemize}
    Using these quantities in the definitions of $\multitotaldegree$ and
    $\multitotalheight$ yield
    \begin{itemize}
        \item $\multitotaldegree = \totaldegree$,
        \item $\multitotalheight \in \widetilde{O}(n\approxheight\totaldegree)$.
    \end{itemize}
    Hence, substituting these values in the bit complexity result of
    \Cref{thm:alg1comp} gives
    \begin{align*}
        &\phantom{{}={}}\widetilde{O}\big(\log(1/\epsproba)n^4\approxheight
        \totaldegree^2(L+nd+n^3) + nd^5\totaldegree^3(n\approxheight\totaldegree
        + n\totaldegree)\big) \\
        &= \widetilde{O}\big(n^2\approxheight\totaldegree^2(
        \log(1/\epsproba)n^2(L+nd+n^3) + d^5\totaldegree^2)\big) \\
        &= \widetilde{O}\big(n^2\totaldegree^2(n\log(1/\epsproba) + n^2 + 
        \tau + d)(\log(1/\epsproba)n^2(L+nd+n^3) + 
        d^5\totaldegree^2)\big)
    \end{align*}
    as a bound for the total number of bit operations of \nameref{algo:smooth}.
    The statement about the height of the output is obtained similarly.
\end{proof}

Finally, if one only wishes to compute the zero-dimensional rational
parametrisations of the points of $S$, without their rational approximation, we
have the following result:

\begin{customthm}{3}\label{thm:smoothparam}
    If one only wishes to obtain the zero-dimensional rational parametrisations
    of the points of $S \coloneqq \left\{\bmx \in \mathbb{R}^n \ | \ f(\bmx)
    \neq 0\right\}$ instead of their rational approximation, this can be done in
    at most
    \[\widetilde{O}\big(\log(1/\epsproba)n^3\multitotaldegree
    \multitotalheight(L+nd+n^3) + nd^3\multitotaldegree^2(\multitotalheight 
    + n\multitotaldegree)\big)\]
    bit operations in the setting of \emph{\Cref{thm:alg1comp}}, or
    \[\widetilde{O}\big(n^2\totaldegree^2(n\log(1/\epsproba) + n^2 + 
    \tau + d)(\log(1/\epsproba)n^2(L+nd+n^3) + 
    d^3\totaldegree)\big)\]
    bit operations in the setting of \emph{\Cref{thm:alg1compsimp}}. There are
    $2n$ such parametrisations computed, consisting of polynomials of degree at
    most $\multitotaldegree$, respectively $\totaldegree$, and height in
    \(\widetilde{O}\big(d^3\multitotaldegree (\multitotalheight +
    n\multitotaldegree)\big)\), respectively \(\widetilde{O}\big(nd^3
    \totaldegree^2 (n\log(1/\epsproba) + n^2 + \tau + d)\big)\).
\end{customthm}
\begin{proof}
    It suffices to apply \nameref{algo:smooth} without computing the
    approximations, that is, returning the $\scrP_k^{\pm}$ once they have been
    computed. The complexity of doing so is thus the same as the one of
    \Cref{thm:alg1comp}, without the complexity of
    \Cref{lin:gen:approx-,lin:gen:approx+}. By \Cref{thm:alg1comp}, this
    therefore gives a bit complexity of
    \begin{align*}
        &\phantom{{}={}}\widetilde{O}\Big(n\big(\log(1/\epsproba)n^2
        \multitotaldegree\multitotalheight(L+nd+n^3) \\
        &\indent \indent + nd^3\multitotaldegree^2 (\multitotalheight + 
        n\multitotaldegree)\big)\Big) \\
        &= \widetilde{O}\big(\log(1/\epsproba)n^3\multitotaldegree
        \multitotalheight(L+nd+n^3) + nd^3\multitotaldegree^2(\multitotalheight 
        + n\multitotaldegree)\big)
    \end{align*}
    in the multi-homogeneous setting. In the case where $f$ has no
    multi-homogeneous structure, substituting $\multitotaldegree = \totaldegree$
    and $\multitotalheight \in \widetilde{O}(n\approxheight\totaldegree)$, as in
    \Cref{thm:alg1compsimp}, yields a bit complexity of
    \[\widetilde{O}\big(n^2\totaldegree^2(n\log(1/\epsproba) + n^2 + 
    \tau + d)(\log(1/\epsproba)n^2(L+nd+n^3) + 
    d^3\totaldegree)\big)\]
    as required. The statements about the heights of the polynomials in the
    output are obtained similarly.
\end{proof}

These results are precisely the bit complexity statements appearing in
\Cref{thm:mainresult,cor:mainresult}.

We remark that the complexities obtained in \Cref{thm:alg1comp} and
\Cref{thm:smoothparam} are, in the worst case scenario, essentially quartic and
cubic, respectively, in the Bézout bound $d(d-1)^{n-1}$, since this quantity is
a bound for $\totaldegree$.

\section{Experiments}\label{sec:exp}

All computations were performed on a single thread of an Intel(R) Xeon(R) Gold
6244 CPU @ 3.60GHz with 1500 GB of memory. We have implemented
\nameref{algo:smooth} in \texttt{SageMath} \cite{sagemath}, using the library
\texttt{msolve} \cite{msolve} for the computation of zero-dimensional rational
parametrisations and real root isolation. The code has been submitted as a
ticket to \texttt{SageMath} to be made publicly available within the software in
a later version\!\footnotemark[1].

\footnotetext[1]{A functional prototype of the code can be found at
\url{https://gillot.perso.lip6.fr/smooth.sage}.}

We compare these computations with our implementation of the algorithm presented
in \cite{SS03,EGS20,EGS23}, which computes points per connected components of
real algebraic sets. When the input polynomial of \nameref{algo:smooth} is $f
\in \ratfield[x_1, \dots, x_n]$, the algorithm of \cite{SS03} applied to
$x_{n+1}f-1$ for a new variable $x_{n+1}$ returns the same information about
sample points of connected components of $S$.
We also compare our results to the Cylindrical Algebraic Decomposition
\cite{collins}, more precisely to the dedicated function \textsc{SamplePoints},
from the \texttt{RegularChains} package present in the computer algebra system
\texttt{Maple} \cite{maple}. 

Finally, we also compare our results to the numerical
solver \textsc{HypersurfaceRegions} \cite{BSW25}, available as a
package in the \texttt{Julia} language \cite{julia}, and to the satisfiability
modulo theory (SMT) solver implemented in the \texttt{SMTLIB} package of
\texttt{Maple}.
These two solvers are fundamentally different to the previous ones. The latter
simply answers the question ``Is $S^+ \coloneqq \{\bmx \in \realfield^n \ | \
f(\bmx) > 0\}$ empty?'', and returns a sample point if it is not; it therefore
computes less information that the other algorithms.
The former computes sample points as our algorithm does, but also determines
whether two sample points lie in the same connected component or not; it
therefore computes more information than our algorithm. However, it is
numerical, and hence does not guarantee that its solutions are accurate, as
shown below: 

\begin{figure}[H]
\centering
\includegraphics{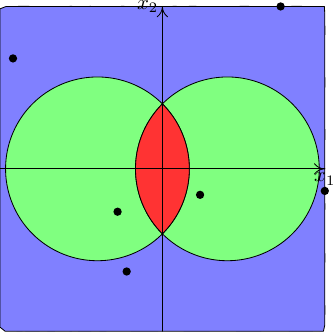}
\caption{Illustration of the output of \textsc{HypersurfaceRegions} applied to
$f = ((x_0-1)^2 + x_1^2 - 1 - 1/2^{16})((x_0+1)^2 + x_1^2 - 1 - 1/2^{16})$.
Despite the success of the computation, the central component (in
\textcolor{red}{red}) is missed (exaggerated size for visual clarity).}
\label{fig:numeric}
\end{figure}

We apply these algorithms to three types of examples:
\begin{itemize}
    \item generic dense polynomials of fixed degree $4$ and height $6$, with
    varying number $n$ of variables,
    \item generic polynomials of fixed degree $12$, with varying number $n$ of
    variables, such that $n-k$ of their partial derivatives have degree $1$, for
    some varying $0 \leq k \leq n$,
    \item and small deformations of singular examples arising from applications
    of other scientific fields, namely \texttt{Vor1} and \texttt{Vor2} from
    computational geometry problems \cite{vor12}, \texttt{P3PGen} from computer
    vision \cite{P3PGen}, \texttt{K1}--\texttt{K4} and
    \texttt{Kalto1}--\texttt{Kalto4} from real algebra theorem-proving
    \cite{K1K2K3K4}, and \texttt{Sot1} and \texttt{Sot2} from problems appearing
    in enumerative geometry \cite{Sot99}.
    These are all polynomials with $4$ to $8$ variables and
    having degrees between $4$ and $18$, of varying density\!\footnotemark[2].
\end{itemize}
\footnotetext[2]{The examples can be found at
\url{https://gillot.perso.lip6.fr/Examples/list.html}.}
The $\infty$ symbol indicates that no result has been obtained after $25$ days
of computation, the \texttt{OOM} symbol indicates memory requirements exceeding
1500 GB, and the \texttt{Err} symbol indicates that the program returned an
error.

Along with computation time, we also compare the total number of points obtained
using the different algorithms. Note that this number may vary between runs for
both the algorithm of \cite{SS03} and \Cref{algo:smooth}, as it depends on the
randomly chosen change of variables matrix $ \changeofvarmatrix$ and
specialisation point $\genericcoordsigma$.

% Rule of thumb for centering in multicolumn: add 2\tabcolsep per sub-column
\newcolumntype{C}[1]{>{\centering\arraybackslash}m{#1}}
\renewcommand{\multirowsetup}{\centering}

\begin{table}[H]
\makebox[1 \textwidth][c]{
\resizebox{1.2 \textwidth}{!}{
\begin{tabular}{|c|C{1.2cm}|C{1.2cm}|C{1.3cm}|C{1.2cm}|C{1.6cm}|C{1.2cm}||C{1.6cm}|C{1.2cm}||C{1.4cm}|}
    \hline
    \multirow{2}{*}{$\bm{n}$} & 
    \multicolumn{2}{C{\dimexpr2.4cm+2\tabcolsep\relax}|}{\textsc{Regular\-Chains.mm}} & 
    \multicolumn{2}{C{\dimexpr2.5cm+2\tabcolsep\relax}|}{\cite[Algorithm 1]{SS03}} & 
    \multicolumn{2}{C{\dimexpr2.8cm+2\tabcolsep\relax}||}{\nameref{algo:smooth}} &
    \multicolumn{2}{C{\dimexpr2.8cm+2\tabcolsep\relax}||}{\textsc{Hypersurface\-Regions.jl}} & 
    \multicolumn{1}{C{1.4cm}|}{\textsc{SMT\-LIB.mm}} \\
    \cline{2-10}
    & Time & Points & Time & Points & Time & Points & Time & Points & Time \\
    \hline
    $4$ & $11 \,000$ & $20 \,607$ & $14$ & $23$ &
    $3$ & $13$ & $84$ & $19$ & $1$ \\
    $5$ & $\infty$ & --- & $500$ & $61$ &
    $10$ & $65$ & $110$ & $33$ & $1$ \\
    $6$ & $\infty$ & --- & $37\, 000$ & $97$ &
    $120$ & $73$ & $230$ & $53$ & $40$ \\
    $7$ & $\infty$ & --- & $\infty$ & --- &
    $8\, 800$ & $177$ & $1\, 100$ & $117$ & $1$ \\
    $8$ & $\infty$ & --- & $\infty$ & --- &
    $1 \,800 \,000$ & $183$ & $7 \,400$ & $113$ & $1$ \\
    $9$ & $\infty$ & --- & $\infty$ & --- & $\infty$ & --- &
    $38 \,000$ & $269$ & $1$ \\
    $10$ & $\infty$ & --- & $\infty$ & --- & $\infty$ & --- &
    $170 \,000$ & $349$ & $430$ \\
    \hline
\end{tabular}
}}
\caption{Comparisons (in seconds) on generic dense polynomials of fixed degree
$4$ and height $6$ with varying number $n$ of variables.}
\label{tab:gen}
\end{table}

In \Cref{tab:gen}, we observe that, on random examples, although
\nameref{algo:smooth} outperforms other symbolic algorithms, it is less
effective than \textsc{HypersurfaceRegions.jl} in dealing with examples of large
dimensions. Of course, unlike other symbolic algorithms,
\textsc{HypersurfaceRegions.jl} does not guarantee that all connected components
have been sampled.
Moreover, as expected, \textsc{SMT\-LIB.mm} easily terminates, since finding a
point of $S^{+}$ (resp.\ $S^-$) is not expected to be a difficult problem when
$f$ is randomly generated.

Furthermore, we observe computationally that, in our implementation of
\nameref{algo:smooth}, solving the polynomial system to compute the critical
points (that is, \Cref{lin:gen:crit}) is the most time-consuming step. This
differs from our complexity analysis, where
\Cref{lin:gen:approx-,lin:gen:approx+}, the solution approximation steps, were
limiting. This can be explained by the fact that the height bounds on computed
quantities, such as the quantity $\lambda$ of \nameref{subr:isol}, are
worst-case bounds that are seldom reached in practice. 

The last case that terminates for our algorithm, $n=8$, requires to solve a
polynomial system of degree $8 748$, and the case $n=9$, which did not
terminate, of degree $26 244$. This degree is known since \texttt{msolve} is
based on multimodular techniques: in the allowed time frame, the system was
solved modulo some prime numbers, but not enough to perform rational
reconstruction, suggesting that speed-ups may be obtained through
multi-threading (as these modular computations are independent).

\begin{table}[H]
\makebox[1 \textwidth][c]{
\resizebox{1.2 \textwidth}{!}{
\begin{tabular}{|c|C{1.2cm}|C{1.2cm}|C{1.4cm}|C{1.2cm}|C{1.6cm}|C{1.2cm}||C{1.6cm}|C{1.2cm}||C{1.4cm}|}
    \hline
    \multirow{2}{*}{$\bm{n}$} & 
    \multicolumn{2}{C{\dimexpr2.4cm+2\tabcolsep\relax}|}{\textsc{Regular\-Chains.mm}} & 
    \multicolumn{2}{C{\dimexpr2.6cm+2\tabcolsep\relax}|}{\cite[Algorithm 1]{SS03}} & 
    \multicolumn{2}{C{\dimexpr2.8cm+2\tabcolsep\relax}||}{\nameref{algo:smooth}} &
    \multicolumn{2}{C{\dimexpr2.8cm+2\tabcolsep\relax}||}{\textsc{Hypersurface\-Regions.jl}} & 
    \multicolumn{1}{C{1.4cm}|}{\textsc{SMT\-LIB.mm}} \\
    \cline{2-10}
    & Time & Points & Time & Points & Time & Points & Time & Points & Time \\
    \hline
    $8,2$ & $260$ & $3 \,105$ & $19 \,000$ & $49$ &
    $24$ & $117$ & $220$ & $41$ & $1$ \\
    $8,3$ & $64 \,000$ & $21 \,343$ & $\infty$ & --- &
    $11 \,000$ & $237$ & $10 \,000$ & $105$ & $1$ \\
    $8,4$ & $\infty$ & --- & $\infty$ & --- &
    $\infty$ & --- & $480 \,000$ & $243$ & $\infty$ \\
    $12,2$ & $20\, 000$ & $25 \,181$ & OOM & --- &
    $31$ & $89$ & $850$ & $37$ & $\infty$ \\
    $12,3$ & $\infty$ & --- & OOM & --- &
    $20 \,000$ & $241$ & $19 \,000$ & $95$ & $\infty$ \\
    $20,2$ & $\infty$ & --- & OOM & --- &
    $110$ & $205$ & $3 \,300$ & $45$ & $1$ \\
    $20,3$ & $\infty$ & --- & OOM & --- &
    $56 \,000$ & $345$ & $140 \, 000$ & 
    \textcolor{red}{Err}\!\footnotemark[4] & $1$ \\
    $50,1$ & $\infty$ & --- & OOM & --- &
    $210$ & $261$ & $5 \, 900$ & $9$ & $1$ \\
    $50,2$ & $\infty$ & --- & OOM & --- &
    $960$ & $325$ & $200 \, 000$ & $43$ & $2$ \\
    $100,0$ & $1$ & $4$ & $22 \,000$ & $247$ &
    $1 \,200$ & $257$ & $15 \,000$ & $13$ & $19$ \\
    $150,0$ & \textcolor{red}{Err}\!\footnotemark[3] & --- &
    $260 \,000$ & $387$ & $7 \,300$ & $397$ & $100 \,000$ &
    17 & 10 \\
    \hline
\end{tabular}
}}
\caption{Comparisons (in seconds) on generic polynomials of fixed degree $12$
and height $6$, with varying number $n$ of variables, such that $n-k$ partial
derivatives have degree $1$. When $k=0$, the polynomial becomes a generic dense
polynomial of degree $2$.}
\label{tab:partial}
\end{table}
\footnotetext[3]{\textsc{RegularChains} returns an error due to a type
\texttt{list} containing too many elements for \texttt{Maple}. We however still
expect the CAD to outperform our algorithm in these degree $2$ examples, both in
computation time and number of sample points, were it not for this type error.}
\footnotetext[4]{This is an intentional error by
\textsc{HypersurfaceRegions.jl}, occurring when it computes critical points
that are numerically too close to being singular points.}

In \Cref{tab:partial}, we have randomly generated polynomials such that most of
their partial derivatives have much lower degree than expected, which is
precisely the case where $\totaldegree$ is much lower than $d(d-1)^{n-1}$, and
hence the case where we expect \nameref{algo:smooth} to perform well, according
to our complexity analysis. Indeed, as can be observed in \Cref{tab:partial},
\nameref{algo:smooth} is competitive with \textsc{HypersurfaceRegions.jl} in
many such examples, and even outperforms \textsc{SMTLIB.mm} in some examples.
\textsc{RegularChains.mm} remains however the most efficient algorithm in
practice when $d=2$.

The one example, $(8,4)$, that we did not manage to compute involved a polynomial
system of degree $29282$. The other examples had degree at most $2662$.

\begin{table}[H]
\makebox[1 \textwidth][c]{
\resizebox{1.2 \textwidth}{!}{
\begin{tabular}{|c|C{1.2cm}|C{1.2cm}|C{1.4cm}|C{1.2cm}|C{1.6cm}|C{1.2cm}||C{1.6cm}|C{1.2cm}||C{1.6cm}|}
    \hline
    \multirow{2}{*}{$\bm{n}$} & 
    \multicolumn{2}{C{\dimexpr2.4cm+2\tabcolsep\relax}|}{\textsc{Regular\-Chains.mm}} & 
    \multicolumn{2}{C{\dimexpr2.6cm+2\tabcolsep\relax}|}{\cite[Algorithm 1]{SS03}} & 
    \multicolumn{2}{C{\dimexpr2.8cm+2\tabcolsep\relax}||}{\nameref{algo:smooth}} &
    \multicolumn{2}{C{\dimexpr2.8cm+2\tabcolsep\relax}||}{\textsc{Hypersurface\-Regions.jl}} & 
    \multicolumn{1}{C{1.6cm}|}{\textsc{SMT\-LIB.mm}} \\
    \cline{2-10}
    & Time & Points & Time & Points & Time & Points & Time & Points & Time \\
    \hline
    \texttt{Vor1} & $4$ & $48$ & $230$ & $47$ & $17$ 
    & 5$3$ & $200$ & \textcolor{red}{Err}\!\footnotemark[4] & $1$ \\
    \texttt{Vor2} & $\infty$ & --- & $830 \, 000$ & $87$ & 
    $17 \, 000$ & $149$ & $220 \, 000$ & $23$ & OOM \\
    \texttt{P3PGen} & $\infty$ & --- & $\infty$ & --- & $80 \, 000$ & $709$ &
    $8 \, 100$ & $234$ & $1$ \\
    \texttt{K1} & OOM & --- & $11 \, 000$ & $643$ & $280$
    & $381$ & $8 \, 400$ & \textcolor{red}{Err}\!\footnotemark[4] & $1$ \\
    \texttt{K2} & $1$ & $112$ & $9 \, 600$ & $693$ & 
    $290$ & $345$ & $3 \, 400$ & \textcolor{red}{Err}\!\footnotemark[4] & $1$ \\
    \texttt{K3} & $\infty$ & --- & $560$ & $355$ & 
    $56$ & $201$ & $1 \, 600$ & \textcolor{red}{Err}\!\footnotemark[4] & $1$ \\
    \texttt{K4} & $2$ & $224$ & $11 \, 000$ & $717$ & 
    $340$ & $373$ & $6 \,700$ & \textcolor{red}{Err}\!\footnotemark[4] & $1$ \\
    \texttt{Kalto1} & $\infty$ & --- &
    $\infty$ & --- & $\infty$ & --- & $\infty$ & --- & $1 \, 100 \, 000$ \\
    \texttt{Kalto2} & $\infty$ & --- &
    $\infty$ & --- & $\infty$ & --- & $\infty$ & --- & $850 \, 000$ \\
    \texttt{Kalto3} & $\infty$ & --- &
    $\infty$ & --- & $\infty$ & --- & $\infty$ & --- & $950 \, 000$ \\
    \texttt{Kalto4} & $\infty$ & --- &
    $\infty$ & --- & $\infty$ & --- & $\infty$ & --- & $1 \, 000 \, 000$ \\
    \texttt{Sot1} & $\infty$ & --- & $\infty$ & --- & $200 \, 000$ & $125$ & $2
    \, 000 \, 000$
    & $47$ & $1$  \\
    \texttt{Sot2} & $\infty$ & --- & $\infty$ & --- & $190 \, 000$ & $125$ &
    $\infty$ & --- & OOM  \\
    \hline
\end{tabular}
}}
\caption{Comparisons (in seconds) on deformed ($1/2^{32}$ subtracted) singular
examples arising from applications.}
\label{tab:sing}
\end{table}

Finally, in \Cref{tab:sing}, we can see that \textsc{HypersurfaceRegions.jl}
fails in most examples due to internal checks, as expected given that the input
polynomials are very close to defining a singular hypersurface. Depending on the
examples, however, either \textsc{RegularChains.mm} or \nameref{algo:smooth}
outperform the other. We also do manage to outperform \textsc{SMTLIB.mm} on some
examples. The highest degree of a polynomial system appearing in an example
which we managed to solve was $10296$.

\paragraph*{Acknowledgements.} The authors are supported by ANR Project
ANR-22-CE91-0007 “EAGLES” and the FA8655-25-1-7469 of the European Office of
Aerospace Research and Development of the Air Force Office of Scientific
Research (AFOSR).

%
% SMTLIB Solver (Maple) timings:
% Vor1:            < 1s
% Vor1_def:        < 1s
% Vor2:            OOM
% Vor2_def:        OOM
% P3PGen:          < 1s
% P3PGen_def:      < 1s
% K1:              < 1s
% K1_def:          < 1s
% K2:              < 1s
% K2_def:          < 1s
% K3:              < 1s
% K3_def:          < 1s
% K4:              < 1s
% K4_def:          < 1s
% Kalto1:          1107271s
% Kalto1_def:      ?
% Kalto2:          850556s
% Kalto2_def:      ?
% Kalto3:          953902s
% Kalto3_def:      ?
% Kalto4:          993124s
% Kalto4_def:      ?
% Gen4var:         < 1s
% Gen5var:         < 1s
% Gen6var:         40s
% Gen7var:         < 1s
% Gen8var:         < 1s
% Gen9var:         < 1s
% Gen10var:        427s
% n8r6:            < 1s
% n8r5:            < 1s
% n8r4:            infty
% n12r10:          infty
% n12r9:           infty
% n20r18:          < 1s
% n20r17:          < 1s
% n50r48:          < 1s
% n50r49:          2s
% n100r100:        19s
% n150r150:        10s
%
%
% Sot1_def:        < 1s
% Sot2_def:        OOM
% BX1_def:         3s
% BX2_def:         194s
% BX3_def:         267s
%
% Other examples:
%
% BX3:               HyRe: 2d1h - 29 points // infty for the rest 
% Sot1_def:            us: 2d7h - 125 points // 22d18h for HyRe - 47 points //
%                          infty for rest
% Sot2_def:            us: 2d6h - 125 points // infty for rest

% \begingroup
% \phantomsection
% \addcontentsline{toc}{chapter}{References}
% \renewcommand{\bibname}{References}
% %\let\clearpage\relax
% \printbibliography
% \endgroup

\end{document}